\documentclass[pdflatex,sn-mathphys-num]{sn-jnl}

\usepackage{graphicx}%
\usepackage{multirow}%
\usepackage{amsmath,amssymb,amsfonts}%
\usepackage{amsthm}%
\usepackage{mathrsfs}%
\usepackage[title]{appendix}%
\usepackage{xcolor}%
\usepackage{textcomp}%
\usepackage{manyfoot}%
\usepackage{booktabs}%
\usepackage{algorithm}%
\usepackage{algorithmicx}%
\usepackage{algpseudocode}%
\usepackage{listings}%
\usepackage{array}
\usepackage{longtable}
\usepackage{makecell}
\usepackage{enumitem}

\theoremstyle{thmstyleone}%
\newtheorem{theorem}{Theorem}
\newtheorem{proposition}[theorem]{Proposition}%
\newtheorem{lemma}[theorem]{Lemma}
\newtheorem{corollary}[theorem]{Corollary}

\theoremstyle{thmstyletwo}%
\newtheorem{example}{Example}%
\newtheorem{remark}{Remark}%

\theoremstyle{thmstylethree}%

\begin{document}

\title[Article Title]{Self-Orthogonal Twisted Generalized Reed-Solomon Codes and Their Application to Quantum Error-Correcting Codes}


\author[1]{\fnm{Yanxin} \sur{Chen}}\email{15581953002@163.com}

\author[1]{\fnm{Yanli} \sur{Wang}}\email{wyl0207111@163.com}

\author*[1]{\fnm{Tongjiang} \sur{Yan}}\email{yantoji@163.com}

\affil[1]{\orgname{China University of Petroleum (East China)}, \orgaddress{\city{Qingdao}, \postcode{266580}, \state{Shandong}, \country{China}}}


\abstract{In this paper, two classes of twisted generalized Reed-Solomon (TGRS) codes with multi-twists are studied. Firstly, some sufficient and necessary conditions for these codes to be self-orthogonal and self-dual are established. Then several explicit constructions of self-orthogonal and self-dual codes are presented, from which quantum stabilizer codes are further derived. Finally, some corresponding examples are given, especially that some of these codes are MDS, AMDS or NMDS and that some of the resulting quantum stabilizer codes are optimal, achieving the quantum Singleton bound.}




\maketitle

\section{Introduction}\label{sec1}

Let $\mathbb{F}_q$ be the finite field with $q=p^h$ elements, where $p$ is a prime and $h$ is a positive integer. Let $\mathbb{F}_q^* = \mathbb{F}_q \setminus \{ 0 \} $ and $\mathbb{F}_q^n$ denote the $n$-dimensional vector space over $\mathbb{F}_q$. Set $S \subseteq \mathbb{F}_q^n $, $\operatorname{Span}(S)$ denotes the linear subspace of $\mathbb{F}_q^n$ spanned by $S$. An $[n,k,d]_q$ linear code $\mathcal{C}$ is a $k$-dimensional linear subspace of $\mathbb{F}_q^n$. The hull of a linear code $\mathcal{C}$ is defined as $ \operatorname{Hull}(\mathcal{C}) = \mathcal{C} \cap \mathcal{C}^\perp $, where $\mathcal{C}^\perp $ is the dual code of $\mathcal{C}$. If $ \operatorname{Hull}(\mathcal{C}) = \mathcal{C} $, then $\mathcal{C} $ is called a self-orthogonal code. If $ \operatorname{Hull}(\mathcal{C}) = \mathcal{C} = \mathcal{C}^\perp $, then $\mathcal{C} $ is called a self-dual code. If $ \operatorname{Hull}(\mathcal{C}) = \{ \boldsymbol{0} \} $, then $\mathcal{C} $ is called a linear complementary dual (LCD) code. In addition, the well-known Singleton bound says that $S(\mathcal{C}) =  n-k+1 - d$ for any $[n,k,d]_q $ code $\mathcal{C}$. If $S(\mathcal{C}) = 0$, then $\mathcal{C}$ is called a maximum distance separable (MDS) code. If $S(\mathcal{C}) = 1$, then $\mathcal{C}$ is called almost-MDS (AMDS) code. If $S(\mathcal{C}) = S(\mathcal{C}^\perp) = 1$, then $\mathcal{C}$ is called near-MDS (NMDS) code.

In recent years, self-orthogonal codes and MDS codes have found many important applications in different fields. Self-orthogonal codes can be used to construct LCD codes \cite{carlet2018euclidean} and various quantum error-correcting codes, such as pure additive quantum codes \cite{steane2002enlargement}, quantum stabilizer codes \cite{ketkar2006nonbinary} and so on, thus providing effective protection for quantum information in quantum computing and quantum communication \cite{calderbank1997quantum,gottesman1996class}. For a given code rate, MDS codes can correct maximum number of errors \cite{huffman2010fundamentals}. Moreover, MDS codes are closely connected with combinatorial designs and finite geometry \cite{huffman2010fundamentals,macwilliams1977theory}. Consequently, the study of MDS codes, including weight distributions, constructions, self-orthogonal property and so on, has been widely investigated \cite{baicheva200410,georgiou2002mds,grassl2008self,guenda2012new,kim2004euclidean,shi2018self}.

In 2017, inspired by twisted Gabidulin codes \cite{sheekey2015new}, Beelen, et al. firstly introduced twisted Reed-Solomon (TRS) codes, and also showed that TRS codes could be well decoded \cite{beelen2017twisted}. Different from generalized Reed-Solomon (GRS) codes, a twisted generalized Reed-Solomon (TGRS) code is not necessarily MDS. And so many scholars have studied the TGRS code, including MDS properties, NMDS properties, LCD properties, self-dual properties, self-orthogonal properties and so on \cite{zhao2025research,huang2023mds,liu2021construction,singh2024mds,meena2025class,liang2025four,sui2022onlymds,gu2023twisted,sui2023new,sui2022mds,zhu2021self,huang2021mds,ding2025new,hu2025p,zhao2025hermitian,guo2023duality,zhang2025almost,zhu20241,zhu2022self,yang2025two,zhu2024class,luo2022two}. In particular, we list some results as follows:

\begin{itemize}
    \item In 2023, Sui et al. provided a sufficient and necessary condition for the TGRS codes with the general 2 twists to be MDS \cite{sui2023new}. In 2025, Ding et al. extended the work of Sui et al. and obtained a sufficient and necessary condition for the TGRS codes with the general $l$ twists to be MDS \cite{ding2025new}. In addition, Zhao et al. presented a sufficient and necessary condition under which an A-TGRS code is MDS by determining the explicit inverse of the Vandermonde matrix \cite{zhao2025research}.

    \item In 2021, Huang et al. presented a sufficient and necessary condition for the (+)-TGRS code to be self-dual, and then constructed several classes of self-dual MDS and NMDS codes \cite{huang2021mds}. In 2023, Sui et al. provided a sufficient and necessary condition for the TGRS codes with the general 2 twists to be self-dual by employing the technique of block matrix, and then constructed several classes of self-dual MDS and NMDS codes \cite{sui2023new}. In 2025, Ding et al. derived a sufficient and necessary condition for TGRS codes with the general $l$ twists to be self-dual using a similar method, and obtained some self-dual TGRS codes \cite{ding2025new}, and then Hu et al. gave a sufficient condition for the $(\mathcal{L}, \mathcal{P})$-TGRS code to be self-dual \cite{hu2025p}.

    \item In 2022, Zhu et al. derived a sufficient and necessary condition for the TGRS code with $h + t \leq k - 1$ to be self-orthogonal by calculating the dual code of the Schur square of the standard twisted Reed-Solomon code \cite{zhu2022self}. In 2024, Zhu et al. presented a sufficient and necessary condition for any punctured code of the $[1,0]$-TGRS code to be self-orthogonal \cite{zhu20241}. In 2025, Yang et al. gave sufficient and necessary conditions for TGRS codes with the hook $h=(0,0)$ and the twists $t=(1,2)$ and $t=(1,3)$ to be MDS and self-orthogonal, respectively, and then constructed two classes of self-orthogonal MDS codes and several classes of LCD codes \cite{yang2025two}. In 2026, Liang et al. established necessary and sufficient conditions for TGRS codes with the hook $h=(k-1,k-1,\cdots,k-1)$ and the twists $t=(1,2,\cdots,l)$ to be NMDS and self-orthogonal, respectively, and then constructed some LCD MDS codes \cite{liang2025multi}.
\end{itemize}

Inspired by these efforts, in this paper, two special classes of $(\mathcal{L},\mathcal{P})$-TGRS codes are considered. We will give sufficient and necessary conditions for TGRS codes to be self-orthogonal. Based on the above results, several classes of self-orthogonal MDS codes and quantum stabilizer codes are constructed via TGRS codes.

The rest of this paper is organized as follows. In Section \ref{sec2}, some basic definitions and results on TGRS codes are presented. In Section \ref{sec3}, sufficient and necessary conditions for the TGRS codes to be self-orthogonal are given. In Section \ref{sec4}, several classes of self-orthogonal codes and quantum stabilizer codes are constructed based on TGRS codes, including self-orthogonal MDS codes and quantum MDS codes. In Section \ref{sec5}, we conclude the whole paper.

\section{Preliminaries}\label{sec2}

For any two vectors $\boldsymbol{x}=(x_1,x_2,\cdots,x_n)$ and $\boldsymbol{y}=(y_1,y_2,\cdots,y_n)$, the Euclidean inner product of $\boldsymbol{x}$ and $\boldsymbol{y}$ is defined as $\boldsymbol{x} \cdot \boldsymbol{y} = \sum \limits_{i=1}^n x_i y_i$. The dual code of a linear code $\mathcal{C}$ is defined to be 
\[
\mathcal{C}^{\perp}=\left\{ \boldsymbol{x} \in \mathbb{F}_q^n \mid\ \boldsymbol{x} \cdot \boldsymbol{c} = \sum_{i=1}^n x_i c_i =0 \text{ for } \boldsymbol{c} \in \mathcal{C} \right\}. 
\]
If $\mathcal{C} \subseteq \mathcal{C}^{\perp}$, then $\mathcal{C}$ is a self-orthogonal code.

For convenience, the following notations are adopted unless otherwise specified.

\begin{itemize}[leftmargin=2.5em]
    \item $\boldsymbol{\alpha} = (\alpha_1,\alpha_2,\cdots,\alpha_n) \in \mathbb{F}_q^n$ with $\alpha_i \neq \alpha_j$ $(i \neq j)$.
    \item $\boldsymbol{v} = (v_1,v_2,\cdots,v_n) \in (\mathbb{F}_q^*)^n$.
    \item $u_i = \prod\limits_{j=1, j \neq i}^n (\alpha_i - \alpha_j)^{-1}$ for $1 \leq i \leq n$.
\end{itemize}

The evaluation map associated with $\boldsymbol{\alpha}$ and $\boldsymbol{v}$ is defined as
\[
\begin{aligned}
ev_{\boldsymbol{\alpha},\boldsymbol{v}} : \mathbb{F}_q[x] &\rightarrow \mathbb{F}_q^n,\\
f(x) &\mapsto (v_1f(\alpha_1),v_2f(\alpha_2),\cdots,v_nf(\alpha_n)).
\end{aligned}
\]

For positive integers $k \leq n \leq q$, denote a $k \times (n-k)$ matrix
\[
A(\eta) = \begin{bmatrix}
\eta_{11} & \eta_{12} & \cdots & \eta_{1,n-k} \\
\eta_{21} & \eta_{22} & \cdots & \eta_{2,n-k} \\
\vdots & \vdots & &\vdots \\
\eta_{k1} & \eta_{k2} & \cdots & \eta_{k,n-k}
\end{bmatrix}
\]
over $\mathbb{F}_q$, then the polynomial set
\[
\mathcal{S} = \left\{ f(x) = \sum_{i=0}^{k-1} f_i x^i + \sum_{i=0}^{k-1} f_{i} \sum_{j=1}^{n-k} \eta_{i+1,j} x^{k-1+j} \mid\ f_i \in \mathbb{F}_q \text{ for } 0 \leq i \leq k-1\right\}
\]
is a $k$-dimensional subspace of $\mathbb{F}_q[x]$ over $\mathbb{F}_q$, and a twisted generalized Reed-Solomon (TGRS) code is defined as $\mathcal{C} = \left\{ ev_{\alpha,v}(f(x)) \mid f(x) \in \mathcal{S}\right\}$.

In this paper, we consider the following two cases:
\[
A_1(\eta) = \begin{bmatrix}
0 & 0 & \cdots & 0 \\
\vdots & \vdots &  &\vdots \\
\eta_{11} & \eta_{12} & \cdots & 0 \\
\eta_{21} & \eta_{22} & \cdots & 0
\end{bmatrix},
A_2(\eta) = \begin{bmatrix}
    \eta_{1} & \eta_{2} & \cdots & \eta_{n-k} \\
    0 & 0 & \cdots & 0 \\
    \vdots & \vdots &  &\vdots \\
    0 & 0 & \cdots & 0
\end{bmatrix}.
\]

Hence we write 

\begin{align*}
    \mathcal{S}_{1} &= \left\{ f(x) = \sum_{i=0}^{k-1} f_i x^i + \sum_{i=1}^2 f_{k-3+i} \sum_{j=1}^2 \eta_{ij} x^{k-1+j} \mid\ f_i \in \mathbb{F}_q \text{ for } 0 \leq i \leq k-1\right\}, \\
    \mathcal{S}_{2} &= \left\{ f(x) = \sum_{i=0}^{k-1} f_i x^i + f_{0} \sum_{j=1}^{n-k} \eta_{j} x^{k-1+j} \mid\ f_i \in \mathbb{F}_q \text{ for } 0 \leq i \leq k-1\right\}.
\end{align*}

Therefore, $\mathcal{C}_1 =\left\{ ev_{\boldsymbol{\alpha},\boldsymbol{v}}(f(x)) \mid f(x) \in \mathcal{S}_1 \right\}$ has the generator matrix
\[
G_1=\begin{bmatrix}
    v_1 & \cdots & v_n \\
    v_1 \alpha_1 & \cdots & v_n \alpha_n \\
    \vdots &  & \vdots \\
    v_1 \alpha_1^{k-3} & \cdots & v_n \alpha_n^{k-3} \\
    v_1(\alpha_1^{k-2} + \sum\limits_{i=1}^2 \eta_{1i} \alpha_1^{k-1+i})
     & \cdots & v_n(\alpha_n^{k-2} + \sum\limits_{i=1}^2 \eta_{1i} \alpha_n^{k-1+i})) \\
    v_1(\alpha_1^{k-1} + \sum\limits_{i=1}^2 \eta_{2i} \alpha_1^{k-1+i})) & \cdots & v_n(\alpha_n^{k-1} + \sum\limits_{i=1}^2 \eta_{2i} \alpha_n^{k-1+i}))
\end{bmatrix}_{k \times n},
\]
and  $\mathcal{C}_2=\left\{ ev_{\boldsymbol{\alpha},\boldsymbol{v}}(f(x)) \mid f(x) \in \mathcal{S}_2 \right\}$ has the generator matrix
\[
G_2=\begin{bmatrix}
    v_1(1+\sum\limits_{i=1}^{n-k}\eta_i \alpha_1^{k-1+i}) & \cdots & v_n(1+\sum\limits_{i=1}^{n-k}\eta_i \alpha_n^{k-1+i}) \\
    v_1 \alpha_1 & \cdots & v_n \alpha_n \\
    \vdots &  & \vdots \\
    v_1 \alpha_1^{k-1} & \cdots & v_n \alpha_n^{k-1}
\end{bmatrix}_{k \times n}.
\]

\begin{lemma}\label{key}\cite{sui2022mds}
    Let $\prod\limits_{i=1}^n(x-\alpha_i) = \sum\limits_{i=0}^n \sigma_i x^{n-i}$. Let $\Lambda_0=1$ and $ \boldsymbol{y}=(\Lambda_0,\Lambda_1,\cdots,\Lambda_n)$ be the unique solution of the following system of equations
\[
\begin{bmatrix}
    \sigma_0 & 0 & 0 & \cdots & 0 \\
    \sigma_1 & \sigma_0 & 0 & \cdots & 0 \\
    \sigma_2 & \sigma_1 & \sigma_0 & \cdots & 0 \\
    \vdots & \vdots & \vdots &  & \vdots \\
    \sigma_{n} & \sigma_{n-1} & \sigma_{n-2} & \cdots & \sigma_{0}
\end{bmatrix}
\begin{bmatrix}
    \Lambda_0 \\ \Lambda_1 \\ \vdots \\ \Lambda_{n} 
\end{bmatrix}
= 
\begin{bmatrix}
    1 \\ 0 \\ \vdots \\ 0
\end{bmatrix}.
\]
For any fixed $0 \leq t \leq n$, if $\alpha_i^{n-1+t} = \sum\limits_{j=0}^{n-1}f_j \alpha_i^j $ for $1 \leq i \leq n$, then $f_{n-1} = \Lambda_{t} $. 
\end{lemma}

\begin{lemma}\cite{sui2023new}
    Let $\prod\limits_{i=1}^n (x - \alpha_i) = \sum\limits_{i=0}^n \sigma_i x^{n-i}$. 
Then $\mathcal{C}_1$ has the parity check matrix
\[
H_1=\begin{bmatrix}
    \cdots & \frac{u_j}{v_j} & \cdots \\
    \cdots & \frac{u_j}{v_j} \alpha_j & \cdots \\
    \vdots & \vdots & \vdots \\
    \cdots & \frac{u_j}{v_j} \alpha_j ^{n-k-3} & \cdots \\
    \cdots & \frac{u_j}{v_j} \alpha_j ^{n-k-2} (1- \sum\limits_{i=1}^{2} \eta_{i2} \sum\limits_{t=0}^{4-i} \sigma_t \alpha_j^{4-i-t}) & \cdots \\
    \cdots & \frac{u_j}{v_j} \alpha_j^{n-k-2} (\sum\limits_{i=0}^1 \sigma_i \alpha_j^{1-i} - \sum\limits_{i=1}^{2} \eta_{i1} \sum\limits_{t=0}^{4-i} \sigma_t \alpha_j^{4-i-t}) & \cdots
\end{bmatrix}_{(n-k) \times n}.
\]
\end{lemma}

\begin{theorem}
    Let $2 \leq k < n $. Let $\eta = (\eta_1 , \cdots , \eta_{n-k}) \in \mathbb{F}_q^{n-k}$ and $\eta_{n-k} \neq 0$. Let $\prod\limits_{i=1}^n (x - \alpha_i) = \sum\limits_{i=0}^n \sigma_i x^{n-i}$. 
Then $\mathcal{C}_2$ has the parity check matrix
\[
H_2=\begin{bmatrix}
    \cdots & \frac{u_j}{v_j} (1-\eta_{n-k}\sum\limits_{i=0}^{n-1}\sigma_{n-1-i}\alpha_j^i) & \cdots\\
    \cdots & \frac{u_j}{v_j} (b_1+\alpha_j) & \cdots \\
    \cdots & \frac{u_j}{v_j} (b_2+\alpha_j^2) & \cdots \\
    \vdots & \vdots & \vdots \\
    \cdots & \frac{u_j}{v_j} (b_{n-k-1}+\alpha_j^{n-k-1}) & \cdots 
\end{bmatrix}_{(n-k) \times n},
\]
where $b_1 = \sigma_1 - \frac{\eta_{n-k-1}}{\eta_{n-k}}$ , $b_i = \sigma_i - \frac{\eta_{n-k-i}}{\eta_{n-k}} - \sum\limits_{j=1}^{i-1} \sigma_{i-j} b_j$ for $ 2 \leq i \leq n-k-1$.
\end{theorem}

\begin{proof}
    To verify that $H_2$ is a parity check matrix of $\mathcal{C}_2$ ,we need to prove that rank($H_2$)=$n-k$ and $G_2 H_2^T = 0$.

Firstly, we show that rank($H_2$)=$n-k$. 

Suppose that $(f_0,f_1,\cdots,f_{n-k-1})$ is a solution of the system of equations
\[(x_0,x_1,\cdots,x_{n-k-1})H_2 = \mathbf{0}.\] 
Let
\[ f(x) = f_0 (1- \eta_{n-k} \sum_{i=0}^{n-1} \sigma_{n-1-i} x^i) + \sum_{i=1}^{n-k-1} f_i (b_i + x^i),\]
where $b_1 = \sigma_1 - \frac{\eta_{n-k-1}}{\eta_{n-k}}$ and $b_i = \sigma_i - \frac{\eta_{n-k-i}}{\eta_{n-k}} - \sum\limits_{j=1}^{i-1} \sigma_{i-j} b_j$ for $ 2 \leq i \leq n-k-1$. Then $f(\alpha_i)=0$ for $1 \leq i \leq n$. But deg($f(x)$)$\leq n-1 < n$, so we have $f(x)=0$. This implies that $f_0 = 0$. Furthermore, we have $f_1 = f_2 = \cdots = f_{n-k-1} = 0$. Therefore, the system of $(x_0,x_1,\cdots,x_{n-k-1})H_2 = \mathbf{0}$ has only the trivial solution. Hence, rank($H_2$)=$n-k$.

In the following, we prove that $G_2H_2^T = \mathbf{0}$.

Denote by $V_n(\alpha)$ an $n \times n$ matrix over $\mathbb{F}_q$ given by 
\[
V_n(\alpha)=\begin{bmatrix}
    1 & 1 & \cdots & 1 \\
    \alpha_1 & \alpha_2 & \cdots & \alpha_n \\
    \vdots & \vdots &   & \vdots \\
    \alpha_1^{n-1} & \alpha_2^{n-1} & \cdots & \alpha_n^{n-1}
\end{bmatrix}.
\]
Consider the system of equations over $\mathbb{F}_q$ given by
\[
V_n(\alpha)(x_1,x_2,\cdots,x_n)^T = (0,0,\cdots,1)^T.
\]
Because $V_n(\alpha)$ is invertible, by the explicit formula for the inverse of the Vandermonde matrix, the system has a unique solution $(u_1,u_2,\cdots,u_n)^T$, where $u_i = \prod\limits_{j=1, j \neq i}^n (\alpha_i - \alpha_j)^{-1}$ for $i=1,2,\cdots,n$. So we obtain that
\[
\left\{
\begin{aligned}
    \sum_{i=1}^n u_i \alpha_i^j &= 0 , \; \text{for } 0 \leq j \leq n-2 , \\
    \sum_{i=1}^n u_i \alpha_i^{n-1} &= 1.
\end{aligned}
\right.
\]

Let $G_2 = \begin{pmatrix}
    g_0 , g_1 , \cdots , g_{k-1}
\end{pmatrix}^T$ and $H_2 = \begin{pmatrix}
    h_0 , h_1 , \cdots ,h_{n-k-1}
\end{pmatrix}^T$, where $g_i$ is the $(i+1)$-th row of $G_2$ and $h_j$ is the $(j+1)$-th row of $H_2$ for all $0 \leq i \leq k-1$ and $0 \leq j \leq n-k-1$. From the above statements, we obtain the following cases. 

(1) It is straightforward to verify that $g_i h_j^T = 0$ for $1 \leq i \leq k-1$, $1 \leq j \leq n-k-1$.

(2) For $1 \leq j \leq n-k-1$, we have
\[ \hspace{2em}
\begin{aligned}
g_0 h_j^T &= \sum_{i=1}^n u_i (1+\sum_{s=1}^{n-k}\eta_s \alpha_i^{k+s-1})(b_j + \alpha_i^j) \\
    &= b_j \sum_{i=1}^n u_i + \sum_{i=1}^n u_i \alpha_i^j +  b_j \sum_{s=1}^{n-k} \eta_{s} \sum_{i=1}^n u_i \alpha_i^{k+s-1} + \sum_{s=1}^{n-k} \eta_s \sum_{i=1}^n u_i \alpha_i^{k+s-1+j} \\
    &= b_j \eta_{n-k} + \sum_{s=n-k-j}^{n-k} \eta_s \sum_{i=1}^n u_i \alpha_i^{k+s-1+j} \\
    &= b_j \eta_{n-k} + \sum_{s=0}^j \eta_{n-k-j+s} \sum_{i=1}^n u_i \alpha_i^{n-1+s}.
\end{aligned}
\]
By Lemma~{\ref{key}}, we have 
\[ \hspace{2em}
\begin{aligned}
    \sum_{i=1}^n u_i \alpha_i^{n-1+t} = \sum_{i=1}^n u_i \sum_{k=0}^{n-1} f_k \alpha_i^k = \sum_{k=0}^{n-1} f_k \sum_{i=1}^n u_i \alpha_i^k = f_{n-1} = \Lambda_t.
\end{aligned}
\]
Using Lemma~{\ref{key}}, it follows that 
\[
\begin{aligned}
g_0 h_j^T &= (\sigma_j - \frac{\eta_{n-k-j}}{\eta_{n-k}} - \sum_{t=1}^{j-1} \sigma_{j-t} b_t) \eta_{n-k} + \sum_{s=0}^j \eta_{n-k-j+s} \Lambda_s \\
&= \eta_{n-k} \sigma_j - \eta_{n-k-j} - \eta_{n-k} \sum_{t=1}^{j-1} \sigma_{j-t} b_t + \sum_{i=0}^j \eta_{n-k-j+i} \Lambda_i \\
&= \sum_{i=1}^j \eta_{n-k-j+i} \sum_{t=0}^i \sigma_t \Lambda_{i-t} \\
&= 0.
\end{aligned}
\]

(3) For $1 \leq i \leq k-1 $, it follows from Lemma~{\ref{key}} that 
\[ \hspace{2em}
\begin{aligned}
g_i h_0^T &= \sum_{j=1}^n u_j \alpha_j^i (1-\eta_{n-k}\sum_{t=0}^{n-1}\sigma_{n-1-t}\alpha_j^t) \\
&= \sum_{j=1}^n u_j \alpha_j^i - \eta_{n-k} \sum_{t=0}^{n-1}\sigma_{n-1-t} \sum_{j=1}^n u_j \alpha_j^{i+t} \\
&= - \eta_{n-k} \sum_{t=n-1-i}^{n-1}\sigma_{n-1-t} \sum_{j=1}^n u_j \alpha_j^{i+t} \\
&= - \eta_{n-k} \sum_{t=0}^i \sigma_{i-t} \sum_{j=1}^n u_j \alpha_j^{n-1+t} \\
&= - \eta_{n-k} \sum_{t=0}^{i} \sigma_{i-t} \Lambda_{t} \\
&= 0.
\end{aligned}
\]

(4) For $i=0$, $j=0$, we have 
\[ \hspace{2em}
\begin{aligned}
g_0 h_0^T =& \sum_{i=1}^n u_i (1+\sum_{s=1}^{n-k}\eta_s \alpha_i^{k+s-1})(1-\eta_{n-k}\sum_{t=0}^{n-1}\sigma_{n-1-t}\alpha_i^t) \\
=& \sum_{i=1}^n u_i - \eta_{n-k} \sum_{t=0}^{n-1}\sigma_{n-1-t} \sum_{i=1}^n u_i \alpha_i^t + \sum_{s=1}^{n-k}\eta_s \sum_{i=1}^n u_i \alpha_i^{k+s-1} \\
&- \eta_{n-k} \sum_{s=1}^{n-k} \eta_s \sum_{t=0}^{n-1} \sigma_{n-1-t} \sum_{i=1}^n u_i \alpha_i^{k+s-1+t} \\
=& - \eta_{n-k} + \eta_{n-k} - \eta_{n-k} \sum_{s=1}^{n-k} \eta_s \sum_{t=0}^{k+s-1} \sigma_{k+s-1-t} \sum_{i=1}^n u_i \alpha_i^{n-1+t} \\
=& -\eta_{n-k} \sum_{s=1}^{n-k} \eta_s \sum_{t=0}^{k+s-1} \sigma_{k+s-1-t} \Lambda_{t} \\
=& 0.
\end{aligned}
\]

Consequently, we conclude that $G_2 H_2^T = \mathbf{0}$, and so $H_2$ is the parity check matrix of $\mathcal{C}_2$.

\end{proof}

\section{The self-orthogonal TGRS Codes}\label{sec3}

\begin{theorem} \label{block}
    Let $\prod\limits_{i=1}^n (x - \alpha_i) = \sum\limits_{i=0}^n \sigma_i x^{n-i}$. Then

\medskip

\noindent \textbf{\textit{Case 1.}} \enspace Let $k \leq \frac{n-4}{2}$. Then $\mathcal{C}_1$ is self-orthogonal if and only if there exists an element $\lambda \in \mathbb{F}_q^*$ satisfying $ v_i^2 = \lambda u_i $ for all $1 \leq i \leq n$.

\medskip

\medskip

\noindent \textbf{\textit{Case 2.}} \enspace Let $n$ be odd and $k=\frac{n-3}{2}$. Then $\mathcal{C}_1$ is self-orthogonal if and only if the two following conditions hold:

\textup{(i)} There exists an element $\lambda \in \mathbb{F}_q^*$ satisfying $ v_i^2 = \lambda u_i $ for all $1 \leq i \leq n$.

\textup{(ii)} $\eta_{12} = \eta_{22} = 0 $.

\medskip

\noindent \textbf{\textit{Case 3.}} \enspace Let $n$ be even and $k=\frac{n-2}{2}$. Then $\mathcal{C}_1$ is self-orthogonal if and only if  the two following conditions hold:

\textup{(i)} There exists an element $\lambda \in \mathbb{F}_q^*$ satisfying $ v_i^2 = \lambda u_i $ for all $1 \leq i \leq n$.

\textup{(ii)} $\eta_{i2}(\eta_{i2} \sigma_1 - 2 \eta_{i1}) = 0 $ for $i = 1,2$, and $ \eta_{12} \eta_{22} \sigma_1 = \eta_{11}\eta_{22} + \eta_{12}\eta_{21} $.

\medskip

\noindent \textbf{\textit{Case 4.}} \enspace Let $n$ be odd and $k=\frac{n-1}{2}$. Then $\mathcal{C}_1$ is self-orthogonal if and only if the two following conditions hold:

\textup{(i)} There exists an element $\lambda \in \mathbb{F}_q^*$ satisfying $ v_i^2 = \lambda u_i $ for all $1 \leq i \leq n$.

\textup{(ii)} $G_1^{\prime} = A^{\prime} H_1^{\prime}$, where 
\[
G_1^{\prime} = \begin{bmatrix}
    \eta_{12} & 0 \\
    \eta_{22} & 0
    \end{bmatrix},
    A^{\prime} = \begin{bmatrix}
        \eta_{12}(\sigma_1^2 - \sigma_2) - \eta_{11} \sigma_1 & \eta_{11} - \eta_{12} \sigma_1 \\
        1 + \eta_{22}(\sigma_1^2 - \sigma_2) - \eta_{21} \sigma_1 & \eta_{21} - \eta_{22} \sigma_1
    \end{bmatrix},
    H_1^{\prime} = \begin{bmatrix}
        - \eta_{22} & - \eta_{12} \\
        - \eta_{21} & - \eta_{11}
    \end{bmatrix}.
\]

\medskip

\noindent \textbf{\textit{Case 5.}} \enspace Let $n$ be even and $k=\frac{n}{2}$. Then $\mathcal{C}_1$ is self-dual if and only if the two following conditions hold:

\textup{(i)} There exists an element $\lambda \in \mathbb{F}_q^*$ satisfying $ v_i^2 = \lambda u_i $ for all $1 \leq i \leq n$.

\textup{(ii)} $G_1^{\prime\prime} = A^{\prime\prime} H_1^{\prime}$, where 
\[
\begin{aligned} 
G_1^{\prime\prime} &= \begin{bmatrix}
    \eta_{11} - \eta_{12} \sigma_1 & \eta_{12} \\
    \eta_{21} - \eta_{22} \sigma_1 & \eta_{22}
\end{bmatrix}, 
H_1^{\prime} = \begin{bmatrix}
        - \eta_{22} & - \eta_{12} \\
        - \eta_{21} & - \eta_{11}
    \end{bmatrix}, \\
A^{\prime\prime} &= \begin{bmatrix}
    1 + \eta_{11}(\sigma_1^2-\sigma_2) + \eta_{12}(2\sigma_1\sigma_2-\sigma_1^3-\sigma_3) & \eta_{12}(\sigma_1^2-\sigma_2) - \eta_{11}\sigma_1 \\
    -\sigma_1 + \eta_{21}(\sigma_1^2-\sigma_2) + \eta_{22}(2\sigma_1\sigma_2-\sigma_1^3-\sigma_3) & 1 + \eta_{22}(\sigma_1^2-\sigma_2) - \eta_{21}\sigma_1
\end{bmatrix}.
\end{aligned}    
\]

\end{theorem}

\begin{proof}

\textbf{\textit{Case 1.}} Since $k \leq \frac{n-4}{2} $, we have $k+1 \leq n-k-3 $. Let $\boldsymbol{\alpha}^{i} = (\alpha_1^i, \alpha_2^i, \cdots, \alpha_n^i)$ and $\boldsymbol{v}\boldsymbol{\alpha}^i = (v_1 \alpha_1^i, v_2 \alpha_2^i, \cdots, v_n \alpha_n^i) $.

The code $\mathcal{C}_1$ is self-orthogonal if and only if each row of $G_1$ can be expressed as a linear combination of the rows of $H_1$. That is, 
\[
\boldsymbol{v}, \boldsymbol{v} \boldsymbol{\alpha}, \cdots, \boldsymbol{v} \boldsymbol{\alpha}^{k-3}, \boldsymbol{v}(\boldsymbol{\alpha}^{k-2} + \sum_{i=1}^2 \eta_{1i} \boldsymbol{\alpha}^{k-1+i}), \boldsymbol{v}(\boldsymbol{\alpha}^{k-1} + \sum_{i=1}^2 \eta_{2i} \boldsymbol{\alpha}^{k-1+i})
\]
can be expressed as linear combinations of
\[
\begin{aligned}
&\frac{\boldsymbol{u}}{\boldsymbol{v}}, \frac{\boldsymbol{u}}{\boldsymbol{v}} \boldsymbol{\alpha}, \cdots, \frac{\boldsymbol{u}}{\boldsymbol{v}} \boldsymbol{\alpha}^{n-k-3}, 
\frac{\boldsymbol{u}}{\boldsymbol{v}} (\boldsymbol{\alpha}^{n-k-2} - \sum_{i=1}^{2} \eta_{i2} \sum_{t=0}^{4-i} \sigma_t \boldsymbol{\alpha}^{n-k+2-i-t}), \\
&\frac{\boldsymbol{u}}{\boldsymbol{v}} (\sum_{i=0}^1 \sigma_i \boldsymbol{\alpha}^{n-k-1-i} - \sum_{i=1}^{2} \eta_{i1} \sum_{t=0}^{4-i} \sigma_t \boldsymbol{\alpha}^{n-k-2-i-t}).
\end{aligned}
\]
Equivalently, this holds if and only if 
\begin{gather*}
\text{Span} \Big\{ \boldsymbol{v}, \boldsymbol{v} \boldsymbol{\alpha}, \cdots, \boldsymbol{v} \boldsymbol{\alpha}^{k-3}, \boldsymbol{v}(\boldsymbol{\alpha}^{k-2} + \sum_{i=1}^2 \eta_{1i} \boldsymbol{\alpha}^{k-1+i}), \boldsymbol{v}(\boldsymbol{\alpha}^{k-1} + \sum_{i=1}^2 \eta_{2i} \boldsymbol{\alpha}^{k-1+i}) \Big\}
\\[4pt]
\subseteq 
\text{Span} \Big\{ \frac{\boldsymbol{u}}{\boldsymbol{v}}, \frac{\boldsymbol{u}}{\boldsymbol{v}} \boldsymbol{\alpha}, \cdots, \frac{\boldsymbol{u}}{\boldsymbol{v}} \boldsymbol{\alpha}^{n-k-3} \Big\}.
\end{gather*}
The condition holds if and only if there exists an element $\lambda \in \mathbb{F}_q^*$ such that $\boldsymbol{v} = \lambda \frac{\boldsymbol{u}}{\boldsymbol{v}} $. It completes the proof of Case 1.

\medskip

\textbf{\textit{Case 2.}} Since $k=\frac{n-3}{2}$, we have 
\begin{align*}
H_1 &= \begin{bmatrix}
    \cdots & \frac{u_j}{v_j} & \cdots \\
    \cdots & \frac{u_j}{v_j} \alpha_j & \cdots \\
    \vdots & \vdots & \vdots \\
    \cdots & \frac{u_j}{v_j} \alpha_j ^{k} & \cdots \\
    \cdots & \frac{u_j}{v_j} \alpha_j ^{k+1} (1- \sum\limits_{i=1}^{2} \eta_{i2} \sum\limits_{t=0}^{4-i} \sigma_t \alpha_j^{4-i-t}) & \cdots \\
    \cdots & \frac{u_j}{v_j} \alpha_j^{k+1} (\sum\limits_{i=0}^1 \sigma_i \alpha_j^{1-i} - \sum\limits_{i=1}^{2} \eta_{i1} \sum\limits_{t=0}^{4-i} \sigma_t \alpha_j^{4-i-t}) & \cdots
\end{bmatrix}_{(n-k) \times n} \\
 &= \begin{bmatrix}
    h_0 \\ h_1 \\ \vdots \\ h_{n-k-1}
\end{bmatrix}.
\end{align*}

Let $G_1 = \begin{pmatrix}
    g_0 , g_1 , \cdots , g_{k-1}
\end{pmatrix}^T$, where $g_i$ is the $(i+1)$-th row of $G_1$ for all $0 \leq i \leq k-1$. The code $\mathcal{C}_1$ is self-orthogonal if and only if each row of $G_1$ can be expressed as a linear combination of the rows of $H_1$. 
Equivalently, this holds if and only if the following two conditions are satisfied:

(a) $\boldsymbol{v}, \boldsymbol{v} \boldsymbol{\alpha}, \cdots, \boldsymbol{v} \boldsymbol{\alpha}^{k-3} $ are linear representations of $\frac{\boldsymbol{u}}{\boldsymbol{v}}, \frac{\boldsymbol{u}}{\boldsymbol{v}} \boldsymbol{\alpha}, \cdots, \frac{\boldsymbol{u}}{\boldsymbol{v}} \boldsymbol{\alpha}^{k-3} $.

(b) The remaining two row vectors of $G_1$ 
\[\boldsymbol{v}(\boldsymbol{\alpha}^{k-2} + \sum_{i=1}^2 \eta_{1i} \boldsymbol{\alpha}^{k-1+i}), \boldsymbol{v}(\boldsymbol{\alpha}^{k-1} + \sum_{i=1}^2 \eta_{2i} \boldsymbol{\alpha}^{k-1+i}) \]
are linear representations of the remaining five row vectors of $H_1$
\[\frac{\boldsymbol{u}}{\boldsymbol{v}} \boldsymbol{\alpha}^{k-2}, \frac{\boldsymbol{u}}{\boldsymbol{v}} \boldsymbol{\alpha}^{k-1} , \frac{\boldsymbol{u}}{\boldsymbol{v}} \boldsymbol{\alpha}^{k}, \frac{\boldsymbol{u}}{\boldsymbol{v}} (\boldsymbol{\alpha}^{k+1} - \sum_{i=1}^{2} \eta_{i2} \sum_{t=0}^{4-i} \sigma_t \boldsymbol{\alpha}^{k+5-i-t}), \]
\[ \frac{\boldsymbol{u}}{\boldsymbol{v}} (\sum_{i=0}^1 \sigma_i \boldsymbol{\alpha}^{k+3-i} - \sum_{i=1}^{2} \eta_{i1} \sum_{t=0}^{4-i} \sigma_t \boldsymbol{\alpha}^{k+5-i-t}). \]
Similarly to the proof of Case 1, the condition (a) holds if and only if $\boldsymbol{v} = \lambda \frac{\boldsymbol{u}}{\boldsymbol{v}} $ for some $\lambda \in \mathbb{F}_q^*$. This establishes the condition (i).

The condition (b) holds if and only if there exists $A_1 = (a_{ij})_{2 \times 5}$ such that
\[ g_{k-3+i} = \sum_{j=1}^{5} a_{ij}h_{n-k-6+j} \] 
for $i=1,2$.

This is equivalent to the following system of linear equations: 
\begin{equation}
\left\{
\begin{aligned}
    a_{i4} \eta_{12} + a_{i5} \eta_{11} &= 0, \\
    a_{i4} (\eta_{22} + \eta_{12} \sigma_1) + a_{i5} (\eta_{21} + \eta_{11} \sigma_1 ) &= 0, \\
    - a_{i4} (\eta_{22} \sigma_1 + \eta_{12} \sigma_2) + a_{i5} (1 - \eta_{21} \sigma_1 - \eta_{11} \sigma_2 ) &= 0, \\
    a_{i4} (1 - \eta_{22} \sigma_2 - \eta_{12} \sigma_3) + a_{i5} (\sigma_1 - \eta_{21} \sigma_2 - \eta_{11} \sigma_3 ) &= \eta_{i2},
\end{aligned}
\right.
\end{equation}
for $i=1,2$.

It follows from System (1) that 
\[
A_1 = \begin{bmatrix}
    1 & 0 & \eta_{11} & \eta_{12} & 0 \\
    0 & 1 & \eta_{21} & \eta_{22} & 0 
\end{bmatrix}
\]
and $\eta_{i2} = 0$ for $i=1,2$.

Therefore, the condition (ii) holds. It completes the proof of Case 2.

\medskip

\noindent \textbf{\textit{Case 3.}} Since $k=\frac{n-2}{2}$, we have 
\begin{align*}
H_1 &=\begin{bmatrix}
    \cdots & \frac{u_j}{v_j} & \cdots \\
    \cdots & \frac{u_j}{v_j} \alpha_j & \cdots \\
    \vdots & \vdots & \vdots \\
    \cdots & \frac{u_j}{v_j} \alpha_j ^{k-1} & \cdots \\
    \cdots & \frac{u_j}{v_j} \alpha_j ^{k} (1- \sum\limits_{i=1}^{2} \eta_{i2} \sum\limits_{t=0}^{4-i} \sigma_t \alpha_j^{4-i-t}) & \cdots \\
    \cdots & \frac{u_j}{v_j} \alpha_j^{k} (\sum\limits_{i=0}^1 \sigma_i \alpha_j^{1-i} - \sum\limits_{i=1}^{2} \eta_{i1} \sum\limits_{t=0}^{4-i} \sigma_t \alpha_j^{4-i-t}) & \cdots
\end{bmatrix}_{(n-k) \times n} \\
 &= \begin{bmatrix}
    h_0 \\ h_1 \\ \vdots \\ h_{n-k-1}
\end{bmatrix}.
\end{align*}

Similarly to the proof of Case 1, the condition (i) holds. 

Let $G_1 = \begin{pmatrix}
    g_0 , g_1 , \cdots , g_{k-1}
\end{pmatrix}^T$, where $g_i$ is the $(i+1)$-th row of $G_1$ for all $0 \leq i \leq k-1$. And the remaining two row vectors of $G_1$ are linear representations of the remaining four row vectors of $H_1$ if and only if there exists $A_2 = (a_{ij})_{2 \times 4}$ such that 
\[ g_{k-3+i} = \sum_{j=1}^{4} a_{ij}h_{n-k-5+j} \]
for $i=1,2$.

This is equivalent to the following system of linear equations: 
\begin{equation}
\left\{
\begin{aligned}
    a_{i3} \eta_{12} + a_{i4} \eta_{11} &= 0, \\
    a_{i3} (\eta_{22} + \eta_{12} \sigma_1) + a_{i4} (\eta_{21} + \eta_{11} \sigma_1 ) &= 0, \\
    - a_{i3} (\eta_{22} \sigma_1 + \eta_{12} \sigma_2) + a_{i4} (1 - \eta_{21} \sigma_1 - \eta_{11} \sigma_2 ) &= \eta_{i2}, \\
    a_{i3} (1 - \eta_{22} \sigma_2 - \eta_{12} \sigma_3) + a_{i4} (\sigma_1 - \eta_{21} \sigma_2 - \eta_{11} \sigma_3 ) &= \eta_{i1},
\end{aligned}
\right.
\end{equation}
for $i=1,2$.

It follows from System (2) that 
\[
A_2 = \begin{bmatrix}
    1 & 0 & \eta_{11}-\eta_{12}\sigma_1 & \eta_{12} \\
    0 & 1 & \eta_{21}-\eta_{22}\sigma_1 & \eta_{22} 
\end{bmatrix}
\]
and 
\[
\left\{
\begin{aligned}
    \eta_{i2}(\eta_{i2} \sigma_1 - 2 \eta_{i1}) &= 0, \; \text{for } i=1,2, \\
    \eta_{11}\eta_{22} + \eta_{12}\eta_{21} &= \eta_{12} \eta_{22} \sigma_1.
\end{aligned}
\right.
\]

Hence, the condition (ii) holds. It completes the proof of Case 3.

\medskip

\noindent \textbf{\textit{Case 4.}} Since $k=\frac{n-1}{2}$, we have 
\begin{align*}
H_1 &=\begin{bmatrix}
    \cdots & \frac{u_j}{v_j} & \cdots \\
    \cdots & \frac{u_j}{v_j} \alpha_j & \cdots \\
    \vdots & \vdots & \vdots \\
    \cdots & \frac{u_j}{v_j} \alpha_j ^{k-2} & \cdots \\
    \cdots & \frac{u_j}{v_j} \alpha_j ^{k-1} (1- \sum\limits_{i=1}^{2} \eta_{i2} \sum\limits_{t=0}^{4-i} \sigma_t \alpha_j^{4-i-t}) & \cdots \\
    \cdots & \frac{u_j}{v_j} \alpha_j^{k-1} (\sum\limits_{i=0}^1 \sigma_i \alpha_j^{1-i} - \sum\limits_{i=1}^{2} \eta_{i1} \sum\limits_{t=0}^{4-i} \sigma_t \alpha_j^{4-i-t}) & \cdots
\end{bmatrix}_{(n-k) \times n} \\
 &= \begin{bmatrix}
    h_0 \\ h_1 \\ \vdots \\ h_{n-k-1}
\end{bmatrix}.
\end{align*}

Similarly to the proof of Case 1, the condition (i) holds. 

Let $G_1 = \begin{pmatrix}
    g_0 , g_1 , \cdots , g_{k-1}
\end{pmatrix}^T$, where $g_i$ is the $(i+1)$-th row of $G_1$ for all $0 \leq i \leq k-1$. And the remaining two row vectors of $G_1$ are linear representations of the remaining three row vectors of $H_1$ if and only if there exists $A_3 = (a_{ij})_{2 \times 3}$ such that
\[ g_{k-3+i} = \sum_{j=1}^{3} a_{ij}h_{n-k-4+j} \]
for $i=1,2$.

This is equivalent to the following system of linear equations: 
\begin{equation}
\left\{
\begin{aligned}
    a_{i2} \eta_{12} + a_{i3} \eta_{11} &= 0 ,  \\
    - a_{i2} (\eta_{22} + \eta_{12} \sigma_1) - a_{i3} (\eta_{21} + \eta_{11} \sigma_1 ) &= \eta_{i2} , \\
    - a_{i2} (\eta_{22} \sigma_1 + \eta_{12} \sigma_2) + a_{i3} (1 - \eta_{21} \sigma_1 - \eta_{11} \sigma_2 ) &= \eta_{i1} , \\
    a_{i2} (1 - \eta_{22} \sigma_2 - \eta_{12} \sigma_3) + a_{i3} (\sigma_1 - \eta_{21} \sigma_2 - \eta_{11} \sigma_3 ) &= i-1 ,
\end{aligned}
\right.
\end{equation}
for $i = 1,2$.

It follows from System (3) that 
\[
A_3 = \begin{bmatrix}
        1 & \eta_{12}(\sigma_1^2 - \sigma_2) - \eta_{11} \sigma_1 & \eta_{11} - \eta_{12} \sigma_1 \\
        0 & 1 + \eta_{22}(\sigma_1^2 - \sigma_2) - \eta_{21} \sigma_1 & \eta_{21} - \eta_{22} \sigma_1
    \end{bmatrix}
\] 
and the following conditions must be satisfied:
\[
\left\{
\begin{aligned}
     (i-1 + \eta_{i2}(\sigma_1^2 - \sigma_2) - \eta_{i1} \sigma_1)\eta_{12} + ( \eta_{i1} - \eta_{i2} \sigma_1 ) \eta_{11} &= 0 , \\
    - (\eta_{i2}(\sigma_1^2 - \sigma_2) - \eta_{i1} \sigma_1) \eta_{22} - ( \eta_{i1} - \eta_{i2} \sigma_1 ) \eta_{21} &= \eta_{i2},
\end{aligned}
\right.
\]
for $i=1,2$.

That is, the condition (ii) holds. It completes the proof of Case 4.

\medskip

\noindent \textbf{\textit{Case 5.}} Since $k=\frac{n}{2}$, we have 
\begin{align*}
H_1 &=\begin{bmatrix}
    \cdots & \frac{u_j}{v_j} & \cdots \\
    \cdots & \frac{u_j}{v_j} \alpha_j & \cdots \\
    \vdots & \vdots & \vdots \\
    \cdots & \frac{u_j}{v_j} \alpha_j ^{k-3} & \cdots \\
    \cdots & \frac{u_j}{v_j} \alpha_j ^{k-2} (1- \sum\limits_{i=1}^{2} \eta_{i2} \sum\limits_{t=0}^{4-i} \sigma_t \alpha_j^{4-i-t}) & \cdots \\
    \cdots & \frac{u_j}{v_j} \alpha_j^{k-2} (\sum\limits_{i=0}^1 \sigma_i \alpha_j^{1-i} - \sum\limits_{i=1}^{2} \eta_{i1} \sum\limits_{t=0}^{4-i} \sigma_t \alpha_j^{4-i-t}) & \cdots
\end{bmatrix}_{(n-k) \times n} \\
 &= \begin{bmatrix}
    h_0 \\ h_1 \\ \vdots \\ h_{n-k-1}
\end{bmatrix}.
\end{align*}

Similarly to the proof of Case 1, the condition (i) holds. 

Let $G_1 = \begin{pmatrix}
    g_0 , g_1 , \cdots , g_{k-1}
\end{pmatrix}^T$, where $g_i$ is the $(i+1)$-th row of $G_1$ for all $0 \leq i \leq k-1$. And the remaining two row vectors of $G_1$ are linear representations of the remaining three row vectors of $H_1$ if and only if there exists $A_4 = (a_{ij})_{2 \times 2}$ such that 
\[ g_{k-3+i} = \sum_{j=1}^{2} a_{ij}h_{n-k-3+j} \]
for $i=1,2$.

This is equivalent to the following system of linear equations: 
\begin{equation}
\left\{
\begin{aligned}
    - a_{i1} \eta_{12} - a_{i2} \eta_{11} &= \eta_{i2} ,  \\
    - a_{i1} (\eta_{22} + \eta_{12} \sigma_1) - a_{i2} (\eta_{21} + \eta_{11} \sigma_1 ) &= \eta_{i1} , \\
    - a_{i1} (\eta_{22} \sigma_1 + \eta_{12} \sigma_2) + a_{i2} (1 - \eta_{21} \sigma_1 - \eta_{11} \sigma_2 ) &= i-1 , \\
    a_{i1} (1 - \eta_{22} \sigma_2 - \eta_{12} \sigma_3) + a_{i2} (\sigma_1 - \eta_{21} \sigma_2 - \eta_{11} \sigma_3 ) &= 2-i ,
\end{aligned}
\right.
\end{equation}
for $i = 1,2$.

It follows from System (4) that 
\[
A_4 = \begin{bmatrix}
        1 + \eta_{11}(\sigma_1^2-\sigma_2) + \eta_{12}(2\sigma_1\sigma_2-\sigma_1^3-\sigma_3) & \eta_{12}(\sigma_1^2-\sigma_2) - \eta_{11}\sigma_1 \\
    -\sigma_1 + \eta_{21}(\sigma_1^2-\sigma_2) + \eta_{22}(2\sigma_1\sigma_2-\sigma_1^3-\sigma_3) & 1 + \eta_{22}(\sigma_1^2-\sigma_2) - \eta_{21}\sigma_1
    \end{bmatrix}
\] 
and the following conditions must be satisfied:
\[
\left\{
\begin{aligned}
    - a_{i1} \eta_{12} - a_{i2} \eta_{11} &= \eta_{i2}, \\
    - a_{i1} \eta_{22} - a_{i2} \eta_{21} &=  \eta_{i1} - \eta_{i2} \sigma_1,
\end{aligned}
\right.
\]
for $i=1,2$.

That is, the condition (ii) holds. It completes the proof of Case 5.
\end{proof}

\begin{remark}
    It should be noted that the result in Case 5 of Theorem~{\ref{block}} is equivalent to the conclusion of Theorem 5.1 in \cite{sui2023new}. Moreover, if $\eta_{11} = \eta_{12} = \eta_{21} = 0$ in Cases 3 and 5 of Theorem~{\ref{block}}, we can obtain Theorem 4.4 (2), 4.3 of \cite{zhang2025almost} as corollaries. If $\eta_{11} = \eta_{12} = \eta_{22} = 0 $ in Cases 4 and 5 of Theorem~{\ref{block}}, we can obtain Theorem 4.1 of \cite{liang2025multi} and Theorem 2.8 of \cite{huang2021mds} as corollaries. If $\eta_{11} = \eta_{12} = 0$ in Cases 2-5 of Theorem~{\ref{block}}, we can obtain Theorem 4.8 (3), 4.5, 4.3, 4.4, 4.8 (1) of \cite{liang2025multi} as corollaries. Hence, we generalize them. More special cases are presented in Table~{\ref{table}}.
\end{remark}

Denote 
\[
\Gamma = \begin{bmatrix}
    \eta_{11} & \eta_{12} \\
    \eta_{21} & \eta_{22}
\end{bmatrix}.
\]

\begin{corollary}\label{block_special}
Assume
    \[
    \Gamma = \begin{bmatrix}
        \eta_1 & 0  \\
        0 & \eta_2 
    \end{bmatrix} \qquad \text{or} \qquad
    \Gamma = \begin{bmatrix}
        0 & \eta_2 \\
        \eta_1 & 0 
    \end{bmatrix},
    \]
    where $\eta_i \in \mathbb{F}_q^{*} $ for $ i = 1,2$. Let $ k = \frac{n-2}{2}$, then $\mathcal{C}_1$ is not self-orthogonal.
\end{corollary}

\begin{proof}
    Assume, to the contrary, $\mathcal{C}_1$ is self-orthogonal. By Case 3 of Theorem~{\ref{block}}, we have $\eta_1 \eta_2 = 0$, which is a contradiction to $\eta_i \in \mathbb{F}_q^{*} $ for $ i = 1,2$. Hence, $\mathcal{C}_1$ is not self-orthogonal.
\end{proof}

\begin{corollary}\label{block_single}
    Assume
    \[
    \Gamma = \begin{bmatrix}
        0 & 0  \\
        0 & \eta
    \end{bmatrix}, \]
    where $\eta \in \mathbb{F}_q^{*} $. Let $ k = \frac{n-3}{2}$, then $\mathcal{C}_1$ is not self-orthogonal.
\end{corollary}

\begin{remark}
    Note that $k < \frac{n}{2} - 1 $ given in Theorem 4.4 (1) of \cite{zhang2025almost} is incorrect. As shown in Corollary~{\ref{block_single}}, $\mathcal{C}_1$ is not self-orthogonal when $k = \frac{n-3}{2}$. Therefore, the correct condition should be $k \leq \frac{n-4}{2}$.
\end{remark}

\begin{theorem}\label{line}
    Let $\prod\limits_{i=1}^n (x - \alpha_i) = \sum\limits_{i=0}^n \sigma_i x^{n-i}$. Then

\medskip

\noindent \textbf{\textit{Case 1.}} \enspace Let $k \leq \frac{n-1}{2} $. Then $\mathcal{C}_2$ is self-orthogonal if and only if the following three conditions are satisfied:

\textup{(i)} There exists an element $\lambda \in \mathbb{F}_q^*$ satisfying $ v_i^2 = \lambda u_i $ for all $1 \leq i \leq n$.

\textup{(ii)} $\eta_{n-k-i} = \eta_{n-k} \sigma_i $ for all $1 \leq i \leq k-1$.

\textup{(iii)} $ \eta_{n-k} \sigma_{n-1} + \sum\limits_{i=1}^{n-2k} (\eta_i - \eta_{n-k} \sigma_{n-k-i}) b_{k+i-1} = 2 $.

\medskip

\noindent \textbf{\textit{Case 2.}} \enspace Let $n = 2k$. Then $\mathcal{C}_2$ is self-dual if and only if the following three conditions are satisfied:

\textup{(i)} There exists an element $\lambda \in \mathbb{F}_q^*$ satisfying $ v_i^2 = \lambda u_i $ for all $1 \leq i \leq n$.

\textup{(ii)} $\eta_{k-i} = \eta_{k} \sigma_i $ for all $1 \leq i \leq k-1$.

\textup{(iii)} $ \eta_{k} \sigma_{n-1} = 2 $.

\end{theorem}

\begin{proof}
    \textbf{\textit{Case 1.}}\enspace Let $G_2 = \begin{pmatrix}
    g_0 , g_1 , \cdots , g_{k-1}
\end{pmatrix}^T$, where $g_i$ is the $(i+1)$-th row of $G_2$ for all $0 \leq i \leq k-1$. The code $\mathcal{C}_2$ is self-orthogonal if and only if each row of $G_2$ can be expressed as a linear combination of the rows of $H_2$.

Firstly, we compare the last $k-1$ rows of $G_2$ with the rows of $H_2$. It follows that the last $k-1$ rows of $G_2$ can be linearly represented by the rows of $H_2$ if and only if there exist $a_{i0},a_{i1},\cdots,a_{i,n-k-1} \in \mathbb{F}_q $, not all zero, such that 
\[ g_i = (a_{i0},a_{i1},\cdots,a_{i,n-k-1})H_2 \]
for $1 \leq i \leq k-1$. Then we can obtain $a_{ij} = 0$ for $j \neq i, 0 \leq j \leq n-k-1$ and $a_{ii} = \boldsymbol{\frac{v^2}{u}}$ for $1 \leq i \leq k-1$. This establishes the condition (i). We can also obtain $b_i = 0$ for $1 \leq i \leq k-1 $. And since $b_1 = \sigma_1 - \frac{\eta_{n-k-1}}{\eta_{n-k}}$ and $b_i = \sigma_i - \frac{\eta_{n-k-i}}{\eta_{n-k}} - \sum\limits_{j=1}^{i-1} \sigma_{i-j} b_j$ for $ 2 \leq i \leq n-k-1$, we have $\eta_{n-k-i} = \eta_{n-k} \sigma_i $ for all $1 \leq i \leq k-1$. This establishes the condition (ii).

Finally, we prove the condition (iii) under the condition $ v_i^2 = \lambda u_i $ for all $1 \leq i \leq n$ and $\lambda \in \mathbb{F}_q^*$.

The first row of $G_2$ is a linear combination of the rows of $H_2$ if and only if there exist $a_0,a_1,\cdots,a_{n-k-1} \in \mathbb{F}_q $, not all zero, such that 
\[ g_0 = (a_0,a_1,\cdots,a_{n-k-1})H_2. \]
Since $k \leq \frac{n-1}{2}$, we have $k \leq n-k-1$. This leads to the following system of linear equations:
\begin{equation}
\left\{
\begin{aligned}
    - a_0 \eta_{n-k} \sigma_i &= \eta_{n-k-i},  &&\text{for } 0 \leq i \leq k-1, \\
    - a_0 \eta_{n-k} \sigma_i + a_{n-i-1} &= \eta_{n-k-i},  &&\text{for } k \leq i \leq n-k-1, \\
    - a_0 \eta_{n-k} \sigma_i + a_{n-i-1} &= 0 ,  &&\text{for } n-k \leq i \leq n-2, \\
    a_0(1-\eta_{n-k} \sigma_{n-1}) + \sum_{i=1}^{n-k-1} a_i b_i &= 1.
\end{aligned}
\right.
\end{equation}
It follows from System (5) that 
\[
\left\{
\begin{aligned}
    a_0 &= -1, \\
    a_i &= - \eta_{n-k} \sigma_{n-i-1}, &&\text{for } 1 \leq i \leq k-1, \\
    a_i &= \eta_{i-k+1} - \eta_{n-k} \sigma_{n-i-1}, &&\text{for } k \leq i \leq n-k-1,
\end{aligned}
\right.
\]
and
\[ \eta_{n-k} \sigma_{n-1} - 1 + \sum_{i=k}^{n-k-1} (\eta_{i-k+1} - \eta_{n-k} \sigma_{n-i-1}) b_i = 1. \] 
That is, the condition (iii) holds. It completes the proof of Case 1.

\medskip

\noindent \textbf{\textit{Case 2.}}\enspace Similarly to the proof of Case 1, the conditions (i) and (ii) hold. It implies $b_i = 0 $ for $1 \leq i \leq k-1$.  

Let $G_2 = \begin{pmatrix}
    g_0 , g_1 , \cdots , g_{k-1}
\end{pmatrix}^T$, where $g_i$ is the $(i+1)$-th row of $G_2$ for all $0 \leq i \leq k-1$. And the first row of $G_2$ is a linear combination of the rows of $H_2$ if and only if there exist $a_0,a_1,\cdots,a_{n-k-1} \in \mathbb{F}_q $, not all zero, such that
\[ g_0 = (a_0,a_1,\cdots,a_{n-k-1})H_2. \]
Since $n=2k$, we have $k = n-k$. This leads to the following system of linear equations:
\begin{equation}
\left\{
\begin{aligned}
    - a_0 \eta_{k} \sigma_i &= \eta_{k-i}, &&\text{ for } 0 \leq i \leq k-1, \\
    - a_0 \eta_{k} \sigma_i + a_{n-i-1} &= 0 , &&\text{ for } k \leq i \leq n-2, \\
    a_0(1-\eta_{k} \sigma_{n-1}) + \sum_{i=1}^{k-1}a_i b_i &= 1.
\end{aligned}
\right.
\end{equation}
It follows from System (6) that 
\[
\left\{
\begin{aligned}
    a_0 &= -1, \\
    a_i &= - \eta_{k} \sigma_{n-i-1}, \; \text{for } 1 \leq i \leq k-1,
\end{aligned}
\right.
\]
and $\eta_{k} \sigma_{n-1} - 1 = 1 $. That is, the condition (iii) holds. It completes the proof of Case 2.

\end{proof}

\section{Construction of Self-Orthogonal TGRS Codes}\label{sec4}

For the convenience of the following description, the definitions of some notations are first given. Let $ \textup{wt}(\mathcal{C}) = \min \{ \text{wt}(\boldsymbol{c}) \mid \boldsymbol{c} \in \mathcal{C}, \boldsymbol{c} \neq \boldsymbol{0} \}$, where $\text{wt} (\boldsymbol{c})$ is the number of nonzero coordinates in $\boldsymbol{c}$. Let $I$ be any $k$-subset of $\{1,2,\cdots,n\}$. We list the following series of symbols:
\begin{align*}
G(x) &= \prod_{i \in I} (x-\alpha_i) = c_0 x^k + c_1 x^{k-1} + \cdots + c_{k-1} x + c_k, \\
A_I &= \begin{bmatrix}
    0 & 1 \\
    0 & 0 & 1 \\
    0 & 0 & 0 & \ddots  \\
    \vdots & \vdots & \vdots & \ddots & 1 \\
    -c_k & -c_{k-1} & -c_{k-2} & \cdots & -c_1
\end{bmatrix}, \\
d_j &= c_{k-j}, a_{mt}^{l} = \sum_{i+j = l, 1\leq i \leq n-k, 0 \leq j \leq t-1} \eta_{mi} d_j, \\
F_{mt}(x) &= \sum_{l=t}^{n-k+t-1} a_{mt}^{l} x^{l-t}, g_{mt} = - \boldsymbol{\gamma} F_{mt}(A_I) \boldsymbol{\gamma}^{T}, \\
M(n,k,\boldsymbol{\alpha},A(\eta),I) &= \begin{vmatrix}
    1 + g_{11} & g_{12} & \cdots & g_{1k} \\
    g_{21} & 1 + g_{22} & \cdots & g_{2k} \\
    \vdots & \vdots & \ddots & \vdots \\
    g_{k1} & g_{k2} & \cdots & 1 + g_{kk}
\end{vmatrix},
\end{align*}
where $\boldsymbol{\gamma} = (0,\cdots,0,1) $, $j=0,1,\cdots,k$, $m, t=1,2,\cdots,k$, $l = 1,2,\cdots,n-1$.

\begin{lemma}\label{MDS}\cite{zhao2025research}
    Suppose that $3 \leq k < n$. Let $\Omega$ be the set of $A(\eta)$ such that $M(n,k,\boldsymbol{\alpha},A(\eta),I) \neq 0 $ for each $k$-subset I of $\{1,2,\cdots,n\}$. Then the code $\mathcal{C} = \left\{ ev_{\alpha,v}(f(x)) \mid f(x) \in \mathcal{S}\right\} $ is MDS if and only if $A(\eta) \in \Omega$.
\end{lemma}

By Lemma~{\ref{MDS}}, we get the following results.

\begin{corollary}\label{MDS_block1}
    Suppose that $3 \leq k < n$. Let 
    \[
    A_1(\eta) = \begin{bmatrix}
        0 & 0 & 0 &\cdots & 0 \\
        \vdots & \vdots & \vdots &  & \vdots \\
        \eta_{11} & \eta_{12} & 0 & \cdots & 0 \\
        \eta_{21} & \eta_{22} & 0 & \cdots & 0
    \end{bmatrix}.
    \]
    Then $\mathcal{C}_1$ is MDS if and only if 
    \[ M(n,k,\boldsymbol{\alpha},A_1(\eta),I) = \begin{vmatrix}
        1+g_{k-1,k-1} & g_{k-1,k} \\
        g_{k,k-1} & 1+g_{kk}
    \end{vmatrix} \neq 0.\]
\end{corollary}

\begin{corollary}
    Suppose that $3 \leq k < n$. Let 
    \[
    A_1(\eta) = \begin{bmatrix}
        0 & 0 & 0 &\cdots & 0 \\
        \vdots & \vdots & \vdots &  & \vdots \\
        0 & 0 & 0 & \cdots & 0 \\
        \eta_1 & \eta_2 & 0 & \cdots & 0
    \end{bmatrix}.
    \]
    Then $\mathcal{C}_1$ is MDS if and only if $\eta_1 c_1 + \eta_2 (c_2 - c_1^2) \neq 1 $.
\end{corollary}

\begin{corollary}
    Suppose that $3 \leq k < n$. Let 
    \[
    A_1(\eta) = \begin{bmatrix}
        0 & 0 & 0 &\cdots & 0 \\
        \vdots & \vdots & \vdots &  & \vdots \\
        0 & \eta_1 & 0 & \cdots & 0 \\
        0 & \eta_2 & 0 & \cdots & 0
    \end{bmatrix}.
    \]
    Then $\mathcal{C}_1$ is MDS if and only if $1 + \eta_1 (c_1 c_2 - c_3) + \eta_2 (c_1^2 - c_2) + \eta_1 \eta_2 c_1 c_2 (c_1^2 + c_2 - \eta_1 c_2 - 1) \neq 0 $.
\end{corollary}

\begin{corollary}
    Suppose that $3 \leq k < n$. Let 
    \[
    A_2(\eta) = \begin{bmatrix}
        \eta_1 & \eta_2 &\cdots & \eta_{n-k} \\
        0 & 0 & \cdots & 0 \\
        \vdots & \vdots &  & \vdots \\
        0 & 0 & \cdots & 0
    \end{bmatrix}.
    \]
    Then $\mathcal{C}_2$ is MDS if and only if $ M(n,k,\boldsymbol{\alpha},A_2(\eta),I) \neq 0 $.
\end{corollary}

\begin{lemma}\cite{ashikhmin2002nonbinary}\label{QECC}
    If there exists a $[n,k,d]_q$ linear code $\mathcal{C}$ with $\mathcal{C} \subseteq \mathcal{C}^{\perp}$, then there exists a $[[n,n-2k, wt(\mathcal{C}^{\perp} \setminus \mathcal{C}) ]]_q$ quantum stabilizer code.
\end{lemma}

\subsection{Construction of Self-Orthogonal TGRS Codes over \texorpdfstring{$\mathbb{F}_q$}{Fq} \texorpdfstring{$(2 \mid q )$}{(2 mid q)}}

\begin{theorem} \label{block1}
Let $s$, $h$ be positive integers and $s \mid h$, $q = 2^h$. Let $\mathbb{F}_{2^s} = \left\{ \alpha_1, \alpha_2, \cdots , \alpha_{2^s} \right\} $ be the subfield of $\mathbb{F}_q $ and $\eta_{11}, \eta_{12}, \eta_{21}, \eta_{22} \in \mathbb{F}_q $ with $\eta_{11} \eta_{22} = \eta_{12} \eta_{21} $. If $\boldsymbol{\alpha} = (\alpha_1, \alpha_2, \cdots, \alpha_{2^s}) $ and $\boldsymbol{v} = (v_1, v_2, \cdots , v_{2^s}) $ with $v_i = \prod\limits_{j=1,j \neq i}^{2^s} (\alpha_i - \alpha_j)^{- 2^{h-1}} $ for $1 \leq i \leq 2^{s}$. Set 
\[
\mathcal{S}_1 = \left\{ \sum_{i=0}^{2^{s-1} - 2} f_i x^i + \sum_{i=1}^2 f_{2^{s-1} - 4 + i} \sum_{j=1}^2 \eta_{ij} x^{2^{s-1} - 2 + j} \mid f_i \in \mathbb{F}_q \text{ for } 0 \leq i \leq 2^{s-1} - 2 \right\}.
\]
Then $\mathcal{C}_1 = ev_{\boldsymbol{\alpha},\boldsymbol{v}}(\mathcal{S}_1)$ is a $[2^{s} , 2^{s-1} - 1 , \geq 2^{s-1}]$ self-orthogonal TGRS code over $\mathbb{F}_q$.
\end{theorem}

\begin{proof}
    On the one hand, since 
    \[u_i = \prod\limits_{j=1,j \neq i}^{2^s} (\alpha_i - \alpha_j)^{-1} = \prod\limits_{j=1,j \neq i}^{2^s} (\alpha_i - \alpha_j)^{-2^{h}} = v_i^2 \] for $i=1,2,\cdots,2^s$, the condition (i) of Case 3 of Theorem~{\ref{block}} holds. 
    On the other hand, write $\prod\limits_{i=1}^{2^s} (x-\alpha_i) = \sum\limits_{i=0}^{2^s} \sigma_i x^{2^s - i} $. Since $\alpha_1, \alpha_2, \cdots,\alpha_{2^s} \in \mathbb{F}_{2^s} $, we have $\prod\limits_{i=1}^{2^s} (x-\alpha_i) = x^{2^s} - x $. Furthermore, we can obtain $\sigma_1 = 0$. Note that $\eta_{11} \eta_{22} = \eta_{12} \eta_{21}$, it implies that the condition (ii) of Case 3 of Theorem~{\ref{block}} holds. Thus, and by Case 3 of Theorem~{\ref{block}}, $\mathcal{C}_1 = ev_{\boldsymbol{\alpha},\boldsymbol{v}}(\mathcal{S}_1)$ is a $[2^s, 2^{s-1}-1]$ self-orthogonal TGRS code over $\mathbb{F}_{2^h}$. Obviously, $\mathcal{C}_1$ is a subcode of $GRS_{2^{s-1}+1}(\boldsymbol{\alpha},\boldsymbol{v})$. Therefore, $\mathcal{C}_1$ has the minimum distance $d \geq 2^{s-1}$. It completes the proof.
\end{proof}

\begin{proposition}
    Let notations be the same as Theorem~{\upshape\ref{block1}}, 
    \[ M(n,k,\boldsymbol{\alpha},A_1(\eta),I) = \begin{vmatrix}
        1+g_{k-1,k-1} & g_{k-1,k} \\
        g_{k,k-1} & 1+g_{kk}
    \end{vmatrix} \neq 0 .\] Then there exists a $[2^{s} , 2^{s-1} - 1 , 2^{s-1}+2]_q$ self-orthogonal MDS TGRS code.
\end{proposition}

By Lemma~{\ref{QECC}} and Theorem~{\ref{block1}}, we get the following proposition.

\begin{proposition}\label{QECC_block1}
    Let notations be the same as Theorem~{\upshape\ref{block1}}. Then there exists a $[[2^{s} , 2 , wt(\mathcal{C}_1^{\perp} \setminus \mathcal{C}_1)]]_q$ quantum stabilizer code.
\end{proposition}

\begin{example}
    Let $s=3$, $h=3$ and $q = 2^3$. Let $\mathbb{F}_{2^3}^{*} = \langle \beta \rangle $, where the minimal polynomial of $\beta$ is $x^3 + x + 1$. Then 
    \[ \boldsymbol{\alpha} = (0,1,\beta,\beta^2,\beta^3,\beta^4,\beta^5,\beta^6) \] and
    \[ \boldsymbol{v} = (1,1,1,1,1,1,1,1). \] 
    Let $\eta_{11} = \eta_{22} = \beta^4, \eta_{12} = 1, \eta_{21} = \beta $, thus we have $ \eta_{11} \eta_{22} = \eta_{12} \eta_{21} $. Then 
    \[
         \mathcal{S}_1 = \left\{ \sum_{i=0}^{2} f_i x^i + \sum_{i=1}^2 f_{i} \sum_{j=1}^2 \eta_{ij} x^{2 + j} \mid f_i \in \mathbb{F}_q \text{ for } 0 \leq i \leq 2 \right\}.
    \]
    According to Theorem~{\ref{block1}}, $\mathcal{C}_1 = ev_{\boldsymbol{\alpha},\boldsymbol{v}}(\mathcal{S}_1)$ is an $[8,3, \geq 4]_{2^3}$ self-orthogonal TGRS code. By Magma, $\mathcal{C}_1 = ev_{\boldsymbol{\alpha},\boldsymbol{v}}(\mathcal{S}_1)$ is an $[8,3, 5]_{2^3}$ self-orthogonal AMDS TGRS code, and the parameters of $\mathcal{C}_1^{\perp}$ are $[8,5,2]_{2^3}$. Applying Proposition~{\ref{QECC_block1}}, we construct a quantum error-correcting code with parameters $[[8,2,2]]_{2^3}$.
\end{example}

\begin{theorem}\label{block2}
    Let $s$, $h$, $t$ be positive integers and $t \mid h$, $s < h$, $q = 2^h $. Assume that $\mathbb{F}_{q}$ is the splitting field of $g(x) = x^{2^s+1} + ax^{2^s} + bx^{2^s-1} +c \in \mathbb{F}_{2^t}[x] $ with $b,c \neq 0$, $\alpha_1, \alpha_2, \cdots, \alpha_{2^s+1}$ are all roots of $g(x)$ and $v_i = (\alpha_i^{2^{s-1}} + b^{2^{h-1}} \alpha_i^{2^{s-1}-1})^{-1} $, $i=1,2,\cdots,2^s+1$. Let $\boldsymbol{\alpha} = (\alpha_1, \alpha_2, \cdots, \alpha_{2^s+1})$, $\boldsymbol{v} = (v_1, v_2, \cdots, v_{2^s+1})$. Let $\eta_{2} \in \mathbb{F}_{q}^{*} $ and $\eta_{1} = \eta_{2}(b^{2^{h-1}}+a) $. Set
    \[
    \mathcal{S}_1 = \left\{ \sum_{i=0}^{2^{s-1}-1} f_i x^i + f_{2^{s-1}-1} \eta_2 \big( (b^{2^{h-1}}+a) x^{2^{s-1}} + x^{2^{s-1}+1} \big) \mid f_i \in \mathbb{F}_q \text{ for } 0 \leq i \leq 2^{s-1}-1 \right\}.
    \]
    Then $\mathcal{C}_1 = ev_{\boldsymbol{\alpha},\boldsymbol{v}}(\mathcal{S}_1)$ is a $[2^s+1, 2^{s-1}, \geq 2^{s-1}]$ self-orthogonal TGRS code over $\mathbb{F}_q$.
\end{theorem}

\begin{proof}
    Since $g'(x)=x^{2^{s}} + bx^{2^s-2} $, we have $\gcd(g(x),g'(x)) = 1 $, then $g(x)$ has $2^s+1$ distinct roots in $\mathbb{F}_q$, namely, $\alpha_i \neq \alpha_j$ $(i \neq j)$. On the one hand, 
    \[u_i = (g'(\alpha_i))^{-1} = ( \alpha_i^{2^s} + b \alpha_i^{2^s-2} )^{-1} = \big( ( \alpha_i^{2^{s-1}} + b^{2^{h-1}} \alpha_i^{2^{s-1}-1} )^{-1} \big)^2 = v_i^2 \] for $1 \leq i \leq 2^s+1$, so the condition (i) of Case 4 of Theorem~{\ref{block}} holds.

    On the other hand, write $g(x) = \prod\limits_{i=1}^{2^s+1} (x-\alpha_i) = \sum\limits_{i=0}^{2^s+1} \sigma_i x^{2^s+1-i} $, we have $\sigma_1 = a$, $\sigma_2 = b$. Hence, 
    \[G_1^{\prime} = \begin{bmatrix}
        0 & 0 \\
        \eta_2 & 0
    \end{bmatrix},
    A^{\prime} = \begin{bmatrix}
        0 & 0 \\
        1 + \eta_2 (a^2+b) + \eta_1 a & \eta_1 + \eta_2 a
    \end{bmatrix}, 
    H_1^{\prime} = \begin{bmatrix}
        \eta_2 & 0 \\
        \eta_1 & 0
    \end{bmatrix}. \]
    Given that $\eta_{1} = \eta_{2}(b^{2^{h-1}}+a) $, it follows that $G_1^{\prime} = A^{\prime} H_1^{\prime} $. That is, the condition (ii) of Case 4 of Theorem~{\ref{block}} holds. Thus, and by Case 4 of Theorem~{\ref{block}}, $\mathcal{C}_1 = ev_{\boldsymbol{\alpha},\boldsymbol{v}}(\mathcal{S}_1)$ is a $[2^s+1, 2^{s-1}]_q$ self-orthogonal TGRS code. Obviously, $\mathcal{C}_1$ is a subcode of $GRS_{2^{s-1}+2}(\boldsymbol{\alpha},\boldsymbol{v})$. Therefore, $\mathcal{C}_1$ has the minimum distance $d \geq 2^{s-1}$. It completes the proof.
\end{proof}

\begin{proposition}\label{MDS_block2}
    Let notations be the same as Theorem~{\upshape\ref{block2}}, $ \eta_1 c_1 + \eta_2 (c_2 - c_1^2) \neq 1 $. Then there exists a $[2^{s}+1 , 2^{s-1} , 2^{s-1}+2]_q$ self-orthogonal MDS TGRS code.
\end{proposition}

By Lemma~{\ref{QECC}} and Theorem~{\ref{block2}}, we get the following proposition.

\begin{proposition}\label{QECC_block2}
    Let notations be the same as Theorem~{\upshape\ref{block2}}. Then there exists a $[[2^{s}+1 , 1 , wt(\mathcal{C}_1^{\perp} \setminus \mathcal{C}_1)]]_q$ quantum stabilizer code.
\end{proposition}

\begin{example}
    Let $s=3$, $t=2$ and $\mathbb{F}_{2^2} = \mathbb{F}_2(\omega)$ with $\omega^2 + \omega + 1 = 0$. Let $g(x) = x^9 + \omega x^7 + \omega^2 \in \mathbb{F}_{2^2}[x] $, then by computer, the splitting field of $g(x)$ over $\mathbb{F}_{2^2}$ is $\mathbb{F}_{2^{16}} $ and $\mathbb{F}_{2^{16}}^{*} = \langle \beta \rangle $, where the minimal polynomial of $\beta$ is $ x^{16} + x^5 + x^3 + x^2 + 1 $. Then the roots of $g(x)$ are $\left\{1,\beta^{18719},\beta^{62609},\beta^{18386},\beta^{53831},\beta^{32036},\beta^{37364},\beta^{9341},\beta^{8009} \right\} $.
    Set 
    \[ \boldsymbol{\alpha} = (1,\beta^{18719},\beta^{62609},\beta^{18386},\beta^{53831},\beta^{32036},\beta^{37364},\beta^{9341},\beta^{8009})   \]  
    and
    \[ \boldsymbol{v} =(\beta^{43690},\beta^{20282},\beta^{42227},\beta^{52883},\beta^{37838},\beta^{59708},\beta^{62372},\beta^{15593},\beta^{14927}) . \]
    Let $\eta_2 = \beta^2 $, $\eta_1 = \beta^{43692} $, then \[
    \mathcal{S}_1 = \left\{ \sum_{i=0}^{3} f_i x^i + \beta^2 f_3(\omega^2 x^4 + x^5) \mid f_i \in \mathbb{F}_{2^{16}} \text{ for } 0 \leq i \leq 3 \right\}.
    \]
    According to Theorem~{\ref{block2}} and Proposition~{\ref{MDS_block2}}, $\mathcal{C}_1 = ev_{\boldsymbol{\alpha},\boldsymbol{v}}(\mathcal{S}_1)$ is a $[9,4,6]_{2^{16}}$ self-orthogonal MDS TGRS code, and the parameters of $\mathcal{C}_1^{\perp}$ are $[9,5,5]_{2^{16}}$. Applying Proposition~{\ref{QECC_block2}}, we construct a quantum error-correcting code with parameters $[[9,1,5]]_{2^5}$, which are quantum MDS codes achieving the quantum Singleton bound.
\end{example}

\begin{theorem}\label{block3}
    Let $s$, $h$, $t$ be positive integers and $t \mid h$, $s \leq h$, $q = 2^h $. Assume that $\mathbb{F}_{q}$ is the splitting field of $g(x) = x^{2^s} + ax^{2^s-1} + bx^{2^s-2} + a^3 x^{2^s-3} + c \in \mathbb{F}_{2^t}[x] $ with $a,b,c \neq 0$, $\alpha_1, \alpha_2, \cdots, \alpha_{2^s}$ are all roots of $g(x)$ and $v_i = (a^{2^{h-1}} \alpha_i^{2^{s-1}-1} + a^{3 \cdot 2^{h-1}} \alpha_i^{2^{s-1}-2})^{-1} $, $i=1,2,\cdots,2^s$. Let $\boldsymbol{\alpha} = (\alpha_1, \alpha_2, \cdots, \alpha_{2^s})$, $\boldsymbol{v} = (v_1, v_2, \cdots, v_{2^s})$. Let $\eta_{1} \in \mathbb{F}_{q}^{*} $ and $\eta_{2} = a \eta_{1} $. Set
    \[
    \mathcal{S}_1 = \left\{ \sum_{i=0}^{2^{s-1}-1} f_i x^i + (f_{2^{s-1}-2} + a f_{2^{s-1}-1}) \eta_1 x^{2^{s-1}+1} \mid f_i \in \mathbb{F}_q \text{ for } 0 \leq i \leq 2^{s-1}-1 \right\}.
    \]
    Then $\mathcal{C}_1 = ev_{\boldsymbol{\alpha},\boldsymbol{v}}(\mathcal{S}_1)$ is a $[2^s, 2^{s-1}, \geq 2^{s-1}-1]$ self-dual TGRS code over $\mathbb{F}_q$.
\end{theorem}

\begin{proof}
    Since $g'(x) = ax^{2^{s}-2} + a^3 x^{2^s-4} $, we have $\gcd(g(x),g'(x)) = 1 $, then $g(x)$ has $2^s$ distinct roots in $\mathbb{F}_q$, namely, $\alpha_i \neq \alpha_j$ $(i \neq j)$. On the one hand,
    \[u_i = (g'(\alpha_i))^{-1} = ( a \alpha_i^{2^s-2} + a^{3} \alpha_i^{2^s-4} )^{-1} = \big( (a^{2^{h-1}} \alpha_i^{2^{s-1}-1} + a^{3 \cdot 2^{h-1}} \alpha_i^{2^{s-1}-2})^{-1} \big)^{2} = v_i^2 \] for $1 \leq i \leq 2^s$, so the condition (i) of Case 5 of Theorem~{\ref{block}} holds.

    On the other hand, write $g(x) = \prod\limits_{i=1}^{2^s} (x-\alpha_i) = \sum\limits_{i=0}^{2^s} \sigma_i x^{2^s-i} $, we have $\sigma_1 = a$, $\sigma_2 = b$ and $\sigma_3 = a^3$. Hence, 
    \[G_1^{\prime\prime} = \begin{bmatrix}
        \eta_1 a & \eta_1 \\
        \eta_2 a & \eta_2
    \end{bmatrix},
    A^{\prime\prime} = \begin{bmatrix}
        1 & \eta_1(a^2+b) \\
        a & 1 + \eta_2(a^2+b)
    \end{bmatrix}, 
    H_1^{\prime} = \begin{bmatrix}
        \eta_2 & \eta_1 \\
        0 & 0
    \end{bmatrix}. \]
    Given that $\eta_{2} = a \eta_{1} $, it follows that $G_1^{\prime\prime} = A^{\prime\prime} H_1^{\prime} $. That is, the condition (ii) of Case 5 of Theorem~{\ref{block}} holds. Thus, and by Case 5 of Theorem~{\ref{block}}, $\mathcal{C}_1 = ev_{\boldsymbol{\alpha},\boldsymbol{v}}(\mathcal{S}_1)$ is a $[2^s, 2^{s-1}]_q$ self-orthogonal TGRS code. Obviously, $\mathcal{C}_1$ is a subcode of $GRS_{2^{s-1}+2}(\boldsymbol{\alpha},\boldsymbol{v})$. Therefore, $\mathcal{C}_1$ has the minimum distance $d \geq 2^{s-1}-1$. It completes the proof.
\end{proof}

\begin{proposition}\label{MDS_block3}
    Let notations be the same as Theorem~{\upshape\ref{block3}}, $ 1 + \eta_1 (c_1 c_2 - c_3) + \eta_2 (c_1^2 - c_2) + \eta_1 \eta_2 c_1 c_2 (c_1^2 + c_2 - \eta_1 c_2 - 1) \neq 0 $. Then there exists a $[2^{s} , 2^{s-1} , 2^{s-1}+1]_q$ self-dual MDS TGRS code.
\end{proposition}

\begin{example}
    Let $s=3$, $t=2$ and $\mathbb{F}_{2^2} = \mathbb{F}_2(\omega)$ with $\omega^2 + \omega + 1 = 0 $. Let $g(x) = x^8 + x^7 + x^6 + x^5 + \omega \in \mathbb{F}_{2^2}[x] $, then by computer, the splitting field of $g(x)$ over $\mathbb{F}_{2^2}$ is $\mathbb{F}_{2^{14}} $ and $\mathbb{F}_{2^{14}}^{*} = \langle \beta \rangle $, where the minimal polynomial of $\beta$ is $x^{14} + x^7 + x^5 + x^3 + 1 $. Then the roots of $g(x)$ are $\left\{ \beta^{5192}, \beta^{1298}, \beta^{1157}, \beta^{8516}, \beta^{2129}, \beta^{4628}, \beta^{4385}, \beta^{10922} \right\} $. 
    Set 
    \[ \boldsymbol{\alpha} = (\beta^{5192}, \beta^{1298}, \beta^{1157}, \beta^{8516}, \beta^{2129}, \beta^{4628}, \beta^{4385}, \beta^{10922}) \]
    and 
    \[\boldsymbol{v}=(\beta^{1910}, \beta^{8669}, \beta^{14177}, \beta^{6263}, \beta^{13853}, \beta^{7559}, \beta^{7640}, \beta^{5461}). \]
    Let $\eta_1 = \eta_2 = \beta^{1504} $, then \[
    \mathcal{S}_1 = \left\{ \sum_{i=0}^{3} f_i x^i + \beta^{1504} ( f_2 + f_3) x^5 \mid f_i \in \mathbb{F}_{2^{14}} \text{ for } 0 \leq i \leq 3 \right\}.
    \]
    According to Theorem~{\ref{block3}} and Proposition~{\ref{MDS_block3}}, $\mathcal{C}_1 = ev_{\boldsymbol{\alpha},\boldsymbol{v}}(\mathcal{S}_1)$ is a $[8,4,5]_{2^{14}} $ self-dual MDS TGRS code.
\end{example}

\begin{theorem}\label{line1}
    Assume that $\mathbb{F}_{q}$ is the splitting field of $g(x) = x^{2^s+1} + \sum\limits_{i=1}^{2^{s-1}}a_i x^{2^s+1-i} + c \in \mathbb{F}_{2^t}[x] $ with $a_i,c \neq 0 $ for $1 \leq i \leq 2^{s-1}$, $\alpha_1, \alpha_2, \cdots, \alpha_{2^s+1}$ are all roots of $g(x)$ and $v_i = (\alpha_i^{2^{s-1}} + \sum\limits_{j=1}^{2^{s-2}} a_{2j}^{2^{h-1}} \alpha_i^{2^{s-1}-j})^{-1} $, $i=1,2,\cdots,2^s+1$. Let $\boldsymbol{\alpha} = (\alpha_1, \alpha_2, \cdots, \alpha_{2^s+1})$, $\boldsymbol{v} = (v_1, v_2, \cdots, v_{2^s+1})$. Let $\eta_{2^{s-1}+1} \in \mathbb{F}_{q}^* $ and $\eta_{i} = \eta_{2^{s-1}+1} a_{2^{s-1}-i+1} $ for $1 \leq i \leq 2^{s-1} $. Set
    \[
    \mathcal{S}_2 = \left\{ \sum_{i=0}^{2^{s-1}-1} f_i x^i + f_0 \eta_{2^{s-1}+1} ( \sum_{j=1}^{2^{s-1}} a_j x^{2^s-j} + x^{2^s})  \mid f_i \in \mathbb{F}_q \text{ for } 0 \leq i \leq 2^{s-1}-1 \right\}.
    \]
    Then $\mathcal{C}_2 = ev_{\boldsymbol{\alpha},\boldsymbol{v}}(\mathcal{S}_2)$ is a $[2^s+1, 2^{s-1}]$ self-orthogonal TGRS code over $\mathbb{F}_q$.
\end{theorem}

\begin{proof}
    Since $g'(x) = x^{2^s} + \sum\limits_{i=1}^{2^{s-2}} a_{2i} x^{2^s-2i} $, we have $\gcd(g(x),g'(x)) = 1$, then $g(x)$ has $2^s+1$ distinct roots in $\mathbb{F}_q$, namely, $\alpha_i \neq \alpha_j$ $(i \neq j)$. On the one hand, 
    \[u_i = (g'(\alpha_i))^{-1} = ( \alpha_i^{2^s} + \sum\limits_{j=1}^{2^{s-2}} a_{2j} \alpha_i^{2^s-2j} )^{-1} = \big( (\alpha_i^{2^{s-1}} + \sum\limits_{j=1}^{2^{s-2}} a_{2j}^{2^{h-1}} \alpha_i^{2^{s-1}-j})^{-1} \big)^2 = v_i^2 \] for $1 \leq i \leq 2^s+1 $, so the condition (i) of Case 1 of Theorem~{\ref{line}} holds.

    On the other hand, write $g(x) = \prod\limits_{i=1}^{2^s+1} (x-\alpha_i) = \sum\limits_{i=0}^{2^s+1} \sigma_i x^{2^s+1-i} $, we have $\sigma_i = a_i$ for $1 \leq i \leq 2^{s-1}$ and $\sigma_{2^s} = 0 $. It is straightforward to verify from the construction of $\eta_{i} $ for $1 \leq i \leq 2^{s-1}+1 $ that $\eta_{2^{s-1}+1-i} = \eta_{2^{s-1}+1} \sigma_i$ for $1 \leq i \leq 2^{s-1}-1 $ and $\eta_{2^{s-1}+1} \sigma_{2^s} + (\eta_1 - \eta_{2^{s-1}+1} \sigma_{2^{s-1}}) b_{2^{s-1}} = 0 $. That is, condition (ii) and (iii) of Case 1 of Theorem~{\ref{line}} hold. Thus, and by Case 1 of Theorem~{\ref{line}}, $\mathcal{C}_2 = ev_{\boldsymbol{\alpha},\boldsymbol{v}}(\mathcal{S}_2)$ is a $[2^s+1, 2^{s-1}]_q$ self-orthogonal TGRS code.
\end{proof}

\begin{proposition}\label{MDS_line1}
    Let notations be the same as Theorem~{\upshape\ref{line1}}, $ M(n,k,\boldsymbol{\alpha},A_2(\eta),I) \neq 0 $. Then there exists a $[2^{s}+1 , 2^{s-1} , 2^{s-1}+2]_q$ self-orthogonal MDS TGRS code.
\end{proposition}

By Lemma~{\ref{QECC}} and Theorem~{\ref{line1}}, we get the following proposition.

\begin{proposition}\label{QECC_line1}
    Let notations be the same as Theorem~{\upshape\ref{line1}}. Then there exists a $[[2^{s}+1 , 1 , wt(\mathcal{C}_2^{\perp} \setminus \mathcal{C}_2)]]_q$ quantum stabilizer code.
\end{proposition}

\begin{example}
    Let $s=3$, $t=1$. Let $g(x) = \sum\limits_{i=0}^{4} x^{9-i} + 1 \in \mathbb{F}_{2}[x] $, then by computer, the splitting field of $g(x)$ over $\mathbb{F}_{2}$ is $\mathbb{F}_{2^{8}} $ and $\mathbb{F}_{2^{8}}^{*} = \langle \beta \rangle $, where the minimal polynomial of $\beta$ is $x^8 + x^4 + x^3 + x^2 + 1 $. Then the roots of $g(x)$ are $\left\{ 1, \beta^{37}, \beta^{148}, \beta^{74}, \beta^{146}, \beta^{164}, \beta^{73}, \beta^{82}, \beta^{41} \right\} $. 
    Set 
    \[ \boldsymbol{\alpha} = (1, \beta^{37}, \beta^{148}, \beta^{74}, \beta^{146}, \beta^{164}, \beta^{73}, \beta^{82}, \beta^{41}) \] 
    and 
    \[ \boldsymbol{v}=(1,\beta^{114},\beta^{201},\beta^{228},\beta^{57},\beta^{78},\beta^{156},\beta^{39},\beta^{147}). \]
    Set $\eta_5 = \beta^{3} $ and $\eta_i = \eta_5 $ for $1 \leq i \leq 4 $, then \[
    \mathcal{S}_2 = \left\{ \sum_{i=0}^{3} f_i x^i + \beta^{3} f_0 ( \sum_{j=1}^4 x^{8-j} + x^8 ) \mid f_i \in \mathbb{F}_{2^{8}} \text{ for } 0 \leq i \leq 3 \right\}.
    \]
    According to the Theorem~{\ref{line1}}, $\mathcal{C}_2 = ev_{\boldsymbol{\alpha},\boldsymbol{v}}(\mathcal{S}_2)$ is a $[9,4]_{2^{8}} $ self-orthogonal TGRS code. By Magma, $\mathcal{C}_2 = ev_{\boldsymbol{\alpha},\boldsymbol{v}}(\mathcal{S}_2)$ is a $[9,4,5]_{2^{8}} $ self-orthogonal NMDS TGRS code, and the parameters of linear code $\mathcal{C}_2^{\perp}$ are $[9,5,4]_{2^8}$.  Applying the Proposition~{\ref{QECC_line1}}, we construct a quantum error-correcting code with parameters $[[9,1,4]]_{2^8}$.
    
    If set $\eta_5 = \beta^7 $ and $\eta_i = \eta_5 $ for $1 \leq i \leq 4 $, then \[
    \mathcal{S}_2 = \left\{ \sum_{i=0}^{3} f_i x^i + \beta^7 f_0 ( \sum_{j=1}^4 x^{8-j} + x^8 ) \mid f_i \in \mathbb{F}_{2^{8}} \text{ for } 0 \leq i \leq 3 \right\}.
    \]
    According to the Theorem~{\ref{line1}} and Proposition~{\ref{MDS_line1}}, $\mathcal{C}_2 = ev_{\boldsymbol{\alpha},\boldsymbol{v}}(\mathcal{S}_2)$ is a $[9,4,6]_{2^{8}} $ self-orthogonal MDS TGRS code and $\mathcal{C}_2^{\perp}$ is a $[9,5,5]$ MDS code. Applying the Proposition~{\ref{QECC_line1}}, we construct a quantum error-correcting code with parameters $[[9,1,5]]_{2^8}$, which are quantum MDS codes achieving the quantum Singleton bound.
    
\end{example}

\begin{theorem}\label{line2}
    Assume that $\mathbb{F}_{q}$ is the splitting field of $g(x) = x^{2^s} + \sum\limits_{i=1}^{2^{s-1}-1}a_i x^{2^s-i} + c \in \mathbb{F}_{2^t}[x] $ with $a_i,c \neq 0 $ for $1 \leq i \leq 2^{s-1}-1$, $\alpha_1, \alpha_2, \cdots, \alpha_{2^s}$ are all roots of $g(x)$ and $v_i = (\sum\limits_{j=1}^{2^{s-2}} a_{2j-1}^{2^{h-1}} \alpha_i^{2^{s-1}-j})^{-1} $, $i=1,2,\cdots,2^s+1$. Let $\boldsymbol{\alpha} = (\alpha_1, \alpha_2, \cdots, \alpha_{2^s})$, $\boldsymbol{v} = (v_1, v_2, \cdots, v_{2^s})$. Let $\eta_{2^{s-1}} \in \mathbb{F}_{q}^* $ and $\eta_{i} = \eta_{2^{s-1}} a_{2^{s-1}-i} $ for $1 \leq i \leq 2^{s-1}-1 $. Set
    \[
    \mathcal{S}_2 = \left\{ \sum_{i=0}^{2^{s-1}-1} f_i x^i + f_0 \eta_{2^{s-1}} ( \sum_{j=1}^{2^{s-1}-1} a_j x^{2^s-1-j} + x^{2^s-1})  \mid f_i \in \mathbb{F}_q \text{ for } 0 \leq i \leq 2^{s-1}-1 \right\}.
    \]
    Then $\mathcal{C}_2 = ev_{\boldsymbol{\alpha},\boldsymbol{v}}(\mathcal{S}_2)$ is a $[2^s, 2^{s-1}]$ self-dual TGRS code over $\mathbb{F}_q$.
\end{theorem}

\begin{proof}
    Since $g'(x) = \sum\limits_{i=1}^{2^{s-2}} a_{2i-1} x^{2^s-2i} $, we have $\gcd(g(x),g'(x)) = 1$, then $g(x)$ has $2^s$ distinct roots in $\mathbb{F}_q$, namely, $\alpha_i \neq \alpha_j$ $(i \neq j)$. On the one hand,
    \[
    u_i = (g'(\alpha_i))^{-1} = (\sum\limits_{j=1}^{2^{s-2}} a_{2j-1} \alpha_{i}^{2^s-2j})^{-1} = \big( (\sum\limits_{j=1}^{2^{s-2}} a_{2j-1}^{2^{h-1}} \alpha_i^{2^{s-1}-j})^{-1} \big)^2 = v_i^2 \]
    for $1 \leq i \leq 2^s $, so the condition (i) of Case 2 of Theorem~{\ref{line}} holds.

    On the other hand, write $g(x) = \prod\limits_{i=1}^{2^s} (x-\alpha_i) = \sum\limits_{i=0}^{2^s} \sigma_i x^{2^s-i} $, we have $\sigma_i = a_i$ for $1 \leq i \leq 2^{s-1}-1$ and $\sigma_{2^s-1} = 0 $. It is straightforward to verify from the construction of $\eta_{i} $ for $1 \leq i \leq 2^{s-1} $ that $\eta_{2^{s-1}-i} = \eta_{2^{s-1}} \sigma_i$ for $1 \leq i \leq 2^{s-1}-1 $ and $\eta_{2^{s-1}} \sigma_{2^s-1} = 0 $. That is, condition (ii) and (iii) of Case 2 of Theorem~{\ref{line}} hold. Thus, and by Case 2 of Theorem~{\ref{line}}, $\mathcal{C}_2 = ev_{\boldsymbol{\alpha},\boldsymbol{v}}(\mathcal{S}_2)$ is a $[2^s, 2^{s-1}]_q$ self-dual TGRS code.
\end{proof}

\begin{proposition}\label{MDS_line2}
    Let notations be the same as Theorem~{\upshape\ref{line2}}, $ M(n,k,\boldsymbol{\alpha},A_2(\eta),I) \neq 0 $. Then there exists a $[2^{s} , 2^{s-1} , 2^{s-1}+1]_q$ self-dual MDS TGRS code.
\end{proposition}

\begin{example}
    Let $s=3$, $t=1$. Let $g(x) = \sum_{i=0}^3 x^{8-i} + 1 \in \mathbb{F}_{2}[x] $, then by computer, the splitting field of $g(x)$ over $\mathbb{F}_{2}$ is $\mathbb{F}_{2^{15}} $ and $\mathbb{F}_{2^{15}}^{*} = \langle \beta \rangle $, where the minimal polynomial of $\beta$ is $x^{15} + x^5 + x^4 + x^2 + 1$. Then the roots of $g(x)$ are $\left\{ \beta^{5285},\beta^{18724},\beta^{9513},\beta^{10570},\beta^{19026},\beta^{9362},\beta^{21140},\beta^{4681}  \right\} $. 
    Set 
    \[ \boldsymbol{\alpha} = (\beta^{5285},\beta^{18724},\beta^{9513},\beta^{10570},\beta^{19026},\beta^{9362},\beta^{21140},\beta^{4681}) \] 
    and 
    \[ \boldsymbol{v} = (
    \beta^{20083},\beta^{4681},\beta^{29596},\beta^{7399},\beta^{26425},\beta^{18724},\beta^{14798},\beta^{9362}). \] 
    Let $\eta_i = \beta^3 $, then 
    \[
    \mathcal{S}_2 = \left\{ \sum_{i=0}^{3} f_i x^i + \beta^3 f_0 \sum_{j=0}^3 x^{7-j} \mid f_i \in \mathbb{F}_{2^{15}} \text{ for } 0 \leq i \leq 3 \right\}.
    \]
    According to the Theorem~{\ref{line2}} and Proposition~{\ref{MDS_line2}}, $\mathcal{C}_2 = ev_{\boldsymbol{\alpha},\boldsymbol{v}}(\mathcal{S}_2)$ is a $[8,4,5]_{2^{15}}$ self-dual MDS TGRS code.
\end{example}

\subsection{Construction of Self-Orthogonal TGRS Codes over \texorpdfstring{$\mathbb{F}_{q}$}{Fq} or \texorpdfstring{$\mathbb{F}_{q^2}$}{Fq2}\texorpdfstring{$(2 \nmid q )$}{(2 not mid q)}}

\begin{theorem}\label{block4}
    Let $s$, $h$, $t$, $r$ be positive integers and $t \mid h$, $s \leq h$, $\gcd(r,p) = 1$. Let $p$ is an odd prime and $q=p^h$. Assume that $\mathbb{F}_q$ is the splitting field of $g(x) = x^{2 r p^s} + bx^{2 r p^s-1} + c $ $(b,c \neq 0)$ over $\mathbb{F}_{p^t}$, $\alpha_1, \alpha_2, \cdots, \alpha_{2 r p^s}$ are all roots of $g(x)$ and $v_i = \alpha_i^{1-rp^s}$, $i=1,2,\cdots,2 r p^s$. Let $\boldsymbol{\alpha} = (\alpha_1, \alpha_2, \cdots, \alpha_{2 r p^s})$, $\boldsymbol{v} = (v_1, v_2, \cdots, v_{2 r p^s})$. Let $\eta_{12}, \eta_{22} \in \mathbb{F}_{q}^{*} $ and $\eta_{i1} = \frac{b}{2} \eta_{i2}$ for $i=1,2$. Set
    \[
    \mathcal{S}_1 = \left\{ \sum_{i=0}^{r p^s-2} f_i x^i + \sum_{i=1}^2 f_{rp^s-4 + i} \eta_{i2} (\frac{b}{2} x^{rp^s-1} + x^{rp^s}) \mid f_i \in \mathbb{F}_q \text{ for } 0 \leq i \leq rp^s-2 \right\}.
    \]
    Then $\mathcal{C}_1 = ev_{\boldsymbol{\alpha},\boldsymbol{v}}(\mathcal{S}_1)$ is a $[2rp^s, rp^s-1, \geq rp^s]$ self-orthogonal TGRS code over $\mathbb{F}_q$.
\end{theorem}

\begin{proof}
    Since $g'(x) = -bx^{2rp^s-2}$, we have $\gcd (g(x),g'(x))=1$, then $g(x)$ has $2rp^s$ distinct roots in $\mathbb{F}_q$, namely, $\alpha_i \neq \alpha_j$ $(i \neq j)$. On the one hand, $u_i = (g'(\alpha_i))^{-1} = -b^{-1} \alpha_i^{2-2rp^s} = -b^{-1} v_i^2 $ for $1 \leq i \leq 2rp^s$, the condition (i) of Case 3 of Theorem~{\ref{block}} holds.

    On the other hand, write $g(x) = \prod\limits_{i=1}^{2rp^s} (x-\alpha_i) = \sum\limits_{i=0}^{2rp^s} \sigma_i x^{2rp^s-i} $, we have $\sigma_1 = b = 2\eta_{11}\eta_{12}^{-1} = 2\eta_{21}\eta_{22}^{-1} $, and hence $2\eta_{i1} = \eta_{i2} \sigma_1 $ for $i=1,2 $ and $\eta_{11} \eta_{22} + \eta_{12} \eta_{21} = \eta_{12} \eta_{22} \sigma_1 $. Therefore, the condition (ii) of Case 3 of Theorem~{\ref{block}} holds. Thus, and by Case 3 of Theorem~{\ref{block}}, $\mathcal{C}_1 = ev_{\boldsymbol{\alpha},\boldsymbol{v}}(\mathcal{S}_1)$ is a $[2rp^s, rp^s-1]_q$ self-orthogonal TGRS code. Obviously, $\mathcal{C}_1$ is a subcode of $GRS_{rp^{s}+1}(\boldsymbol{\alpha},\boldsymbol{v})$. Therefore, $\mathcal{C}_1$ has the minimum distance $d \geq rp^s$.
\end{proof}

\begin{proposition}\label{MDS_block4}
    Let notations be the same as Theorem~{\upshape\ref{block4}}, 
    \[ M(n,k,\boldsymbol{\alpha},A_1(\eta),I) = \begin{vmatrix}
        1+g_{k-1,k-1} & g_{k-1,k} \\
        g_{k,k-1} & 1+g_{kk}
    \end{vmatrix} \neq 0 .\]
    Then there exists a $[2rp^s, rp^s-1, rp^s+2]_q$ self-orthogonal MDS TGRS code.
\end{proposition}

By Lemma~{\ref{QECC}} and Theorem~{\ref{block4}}, we get the following proposition.

\begin{proposition}\label{QECC_block4}
    Let notations be the same as Theorem~{\upshape\ref{block4}}. Then there exists a $[[2rp^{s} , 2, wt(\mathcal{C}_1^{\perp} \setminus \mathcal{C}_1)]]_q$ quantum stabilizer code.
\end{proposition}

\begin{example}
    Let $s=1$, $t=1$ and $r=2$. Let $g(x) = x^{12} + 2x^{11} + 2 \in \mathbb{F}_{3}[x] $, then by computer, the splitting field of $g(x)$ over $\mathbb{F}_{3}$ is $\mathbb{F}_{3^8} $ and $\mathbb{F}_{3^8}^{*} = \langle \beta \rangle $, where the minimal polynomial of $\beta$ is $x^8 + 2x^5 + x^4 + 2x^2 + 2x + 2 $. Then the roots of $g(x)$ are $\left\{\beta^{1713},\beta^{5139},\beta^{571},\beta^{3444},\beta^{2377},\beta^{1148},\beta^{2979},\beta^{2297},\beta^{331},\beta^{4756},\beta^{993},\beta^{3772}
     \right\} $. 
    Set 
    \[ \boldsymbol{\alpha} = (\beta^{1713},\beta^{5139},\beta^{571},\beta^{3444},\beta^{2377},\beta^{1148},\beta^{2979},\beta^{2297},\beta^{331},\beta^{4756},\beta^{993},\beta^{3772}) \] 
    and 
    \[ \boldsymbol{v} =(\beta^{4555},\beta^{545},\beta^{3705},\beta^{2460},\beta^{1235},\beta^{820},\beta^{4785},\beta^{1635},\beta^{4905},\beta^{2460},\beta^{1595},\beta^{820}) . \]
    Let $\eta_{11} = \eta_{12} = \beta^{2232}, \eta_{21} = \eta_{22} = \beta^{2304} $, then \[
    \mathcal{S}_1 = \left\{ \sum_{i=0}^{4} f_i x^i + \sum_{i=1}^2 f_{i+2} \eta_{i2} ( x^5 +x^6 ) \mid f_i \in \mathbb{F}_{3^{8}} \text{ for } 0 \leq i \leq 4 \right\}.
    \]
    According to the Theorem~{\ref{block4}} and Proposition~{\ref{MDS_block4}}, $\mathcal{C}_1 = ev_{\boldsymbol{\alpha},\boldsymbol{v}}(\mathcal{S}_1)$ is a $[12,5,8]_{3^{8}} $ self-orthogonal MDS TGRS code and $\mathcal{C}_1^{\perp}$ is a $[12,7,6]_{3^6}$ MDS code. Applying the Proposition~{\ref{QECC_block4}}, we construct a quantum error-correcting code with parameters $[[6,2,3]]_{3^6} $, which are quantum MDS codes achieving the quantum Singleton bound.
\end{example}

\begin{theorem}\label{block5}
    Let $s$, $h$, $t$ be positive integers, $r$ is odd and $t \mid h$, $s \leq h$, $\gcd(r,p) = 1$. Let $p$ is an odd prime and $q=p^h$. Assume that $\mathbb{F}_q$ is the splitting field of $g(x) = x^{rp^s} + a x^{rp^s-1} + b x^{rp^s-2} + c \in \mathbb{F}_{p^t}[x] $ with $a,b,c \neq 0 $, $\alpha_1, \alpha_2, \cdots, \alpha_{rp^s}$ are all roots of $g(x)$. There exist $v_i \in \mathbb{F}_{q^2}$ such that $v_i^2 = (g'(\alpha_i))^{-1} $ for $1 \leq i \leq {rp^s}$. Let $\boldsymbol{\alpha} = (\alpha_1, \alpha_2, \cdots, \alpha_{rp^s})$, $\boldsymbol{v} = (v_1, v_2, \cdots, v_{rp^s})$. Let $\eta_{2} = 2b^{-1} $ and $\eta_{1} = 2ab^{-1} $. Set
    \[
    \mathcal{S}_1 = \left\{ \sum_{i=0}^{\frac{rp^s-1}{2}-1} f_i x^i + \frac{2}{b} f_{\frac{rp^s-1}{2}-1} \big(a x^{\frac{rp^s-1}{2}} + b x^{\frac{rp^s-1}{2}+1} \big) \mid f_i \in \mathbb{F}_{q^2} \text{ for } 0 \leq i \leq {\frac{rp^s-1}{2}-1} \right\}.
    \]
    Then $\mathcal{C}_1 = ev_{\boldsymbol{\alpha},\boldsymbol{v}}(\mathcal{S}_1)$ is a $[rp^s, \frac{rp^s-1}{2}, \geq \frac{rp^s-1}{2}]$ self-orthogonal TGRS code over $\mathbb{F}_{q^2}$.
\end{theorem}

\begin{proof}
    Since $g'(x) = -ax^{rp^s-2} - 2bx^{rp^s-3} $, we have $\gcd(g(x),g'(x)) = 1 $, then $g(x)$ has $rp^s$ distinct roots in $\mathbb{F}_q$, namely, $\alpha_i \neq \alpha_j$ $(i \neq j)$. On the one hand, $u_i = (g'(\alpha_i))^{-1} \in \mathbb{F}_q$ is a square elements of $\mathbb{F}_{q^2}$, then there exist $v_i \in \mathbb{F}_{q^2}$ such that $u_i = v_i^2$ for $1 \leq i \leq rp^s$, the condition (i) of Case 4 of Theorem~{\ref{block}} holds.

    On the other hand, write $g(x) = \prod\limits_{i=1}^{rp^s} (x-\alpha_i) = \sum\limits_{i=0}^{rp^s} \sigma_i x^{rp^s-i} $, we have $\sigma_1 = a$, $\sigma_2 = b$. Hence, 
    \[G_1^{\prime} = \begin{bmatrix}
        0 & 0 \\
        \eta_2 & 0
    \end{bmatrix},
    A^{\prime} = \begin{bmatrix}
        0 & 0 \\
        1 + \eta_2 (a^2-b) + \eta_1 a & \eta_1 - \eta_2 a
    \end{bmatrix}, 
    H_1^{\prime} = \begin{bmatrix}
        -\eta_2 & 0 \\
        -\eta_1 & 0
    \end{bmatrix}.\] 
    Given that $\eta_{2} = 2b^{-1} $ and $\eta_{1} = 2ab^{-1} $, it follows that $G_1^{\prime} = A^{\prime} H_1^{\prime} $. That is, the condition (ii) of Case 4 of Theorem~{\ref{block}} holds. Thus, and by Case 4 of Theorem~{\ref{block}}, $\mathcal{C}_1 = ev_{\boldsymbol{\alpha},\boldsymbol{v}}(\mathcal{S}_1)$ is a $[rp^s, \frac{rp^s-1}{2}]_{q^2}$ self-orthogonal TGRS code. Obviously, $\mathcal{C}_1$ is a subcode of $GRS_{\frac{rp^s+3}{2}}(\boldsymbol{\alpha},\boldsymbol{v})$. Therefore, $\mathcal{C}_1$ has the minimum distance $d \geq \frac{rp^s-1}{2}$. It completes the proof.
\end{proof}

\begin{proposition}\label{MDS_block5}
    Let notations be the same as Theorem~{\upshape\ref{block5}}, $ \eta_1 c_1 + \eta_2 (c_2 - c_1^2) \neq 1 $. Then there exists a $[rp^s, \frac{rp^s-1}{2}, \frac{rp^s+3}{2}]_{q^2}$ self-orthogonal MDS TGRS code.
\end{proposition}

By Lemma~{\ref{QECC}} and Theorem~{\ref{block5}}, we get the following proposition.

\begin{proposition}\label{QECC_block5}
    Let notations be the same as Theorem~{\upshape\ref{block5}}. Then there exists a $[[rp^{s} , 1 , wt(\mathcal{C}_1^{\perp} \setminus \mathcal{C}_1)]]_{q^2}$ quantum stabilizer code.
\end{proposition}

\begin{example}
    Let $s=1$, $t=1$ and $r=1$. Let $g(x) = x^{5} + x^4 + x^3 + 1 \in \mathbb{F}_{5}[x] $, then by computer, the splitting field of $g(x)$ over $\mathbb{F}_{5}$ is $\mathbb{F}_{5^4} $ and $\mathbb{F}_{5^4}^{*} = \langle \beta \rangle $, where the minimal polynomial of $\beta$ is $x^4 + 4x^2 + 4x + 2 $. Let $\mathbb{F}_{5^8}^{*} = \langle \gamma \rangle $, where the minimal polynomial of $\gamma$ is $x^8 + x^4 + 3x^2 + 4x + 2 $. Then the roots of $g(x)$ are $\left\{ 4, \beta^{512}, \beta^{64}, \beta^{352}, \beta^{320} \right\} $. Set 
    \[ \boldsymbol{\alpha} = (4, \beta^{512}, \beta^{64}, \beta^{352}, \beta^{320}) \] 
    and 
    \[ \boldsymbol{v} =(\gamma^{97656},\gamma^{89205},\gamma^{250713},\gamma^{213153},\gamma^{81693} ) .\]
    Let $\eta_{1} = \eta_{2} = 2 $, then \[
    \mathcal{S}_1 = \left\{ \sum_{i=0}^{1} f_i x^i + 2 f_1 ( x^2 + x^3 ) \mid f_i \in \mathbb{F}_{5^8} \text{ for } 0 \leq i \leq 1 \right\}.
    \]
    According to the Theorem~{\ref{block5}}, $\mathcal{C}_1 = ev_{\boldsymbol{\alpha},\boldsymbol{v}}(\mathcal{S}_1)$ is a $[5,2,\geq 2]_{5^{8}} $ self-orthogonal TGRS code. By Magma, $\mathcal{C}_1 = ev_{\boldsymbol{\alpha},\boldsymbol{v}}(\mathcal{S}_1)$ is a $[5,2,4]_{5^{8}} $ self-orthogonal MDS TGRS code and $\mathcal{C}_1^{\perp} $ is a $[5,3,3]_{5^8}$ MDS code. Applying the Proposition~{\ref{QECC_block5}}, we construct a quantum error-correcting code with parameters $[[5,1,3]]_{5^8} $, which are quantum MDS codes achieving the quantum Singleton bound.
\end{example}

\begin{theorem}\label{block6}
    Let $s$, $h$, $t$, $r$ be positive integers and $t \mid h$, $s \leq h$.  Let $p$ is an odd prime, $q=p^h$. Assume that $\mathbb{F}_q$ is the splitting field of $g(x) = x^{2rp^s} + a x^{2rp^s-1} + b x^{2rp^s-2} + c \in \mathbb{F}_{p^t}[x] $ with $a,b,c \neq 0$ and $b \neq \frac{1}{2}a^2$, $\alpha_1, \alpha_2, \cdots, \alpha_{2rp^s}$ are all roots of $g(x)$. There exist $v_i \in \mathbb{F}_{q^2}$ such that $v_i^2 = (g'(\alpha_i))^{-1} $ for $1 \leq i \leq {2rp^s}$. Let $\boldsymbol{\alpha} = (\alpha_1, \alpha_2, \cdots, \alpha_{2rp^s})$, $\boldsymbol{v} = (v_1, v_2, \cdots, v_{2rp^s})$. Let $\eta_1 = 2 a^{-1} (a^2 - 2b)^{-1} $ and $\eta_2 = - a \eta_1 $. Set
    \[
    \mathcal{S}_1 = \left\{ \sum_{i=0}^{rp^s-1} f_i x^i + 2 (a^2 - 2b)^{-1} (a^{-1} f_{rp^s-2} - f_{rp^s-1}) x^{rp^s+1} \mid f_i \in \mathbb{F}_{q^2} \text{ for } 0 \leq i \leq rp^s-1 \right\}.
    \]
    Then $\mathcal{C}_1 = ev_{\boldsymbol{\alpha},\boldsymbol{v}}(\mathcal{S}_1)$ is a $[2rp^s, rp^s, \geq rp^s-1]$ self-dual TGRS code over $\mathbb{F}_{q^2}$.
\end{theorem}

\begin{proof}
    Since $g'(x) = -ax^{2rp^s-2} - 2b x^{2rp^s-3} $, we have $\gcd(g(x),g'(x)) = 1 $, then $g(x)$ has $2rp^s$ distinct roots in $\mathbb{F}_q$, namely, $\alpha_i \neq \alpha_j$ $(i \neq j)$. On the one hand, $u_i = (g'(\alpha_i))^{-1} \in \mathbb{F}_q$ is a square elements of $\mathbb{F}_{q^2}$, then there exist $v_i \in \mathbb{F}_{q^2}$ such that $u_i = v_i^2$, so the condition (i) of Case 5 of Theorem~{\ref{block}} holds.

    On the other hand, write $g(x) = \prod\limits_{i=1}^{2rp^s} (x-\alpha_i) = \sum\limits_{i=0}^{2rp^s} \sigma_i x^{2rp^s-i} $, we have $\sigma_1 = a$, $\sigma_2 = b$. Hence, 
    \[G_1^{\prime\prime} = \begin{bmatrix}
        -\eta_1 a & \eta_1 \\
        -\eta_2 a & \eta_2
    \end{bmatrix},
    A^{\prime\prime} = \begin{bmatrix}
        1 + \eta_1 (2ab - a^3) & \eta_1(a^2-b) \\
        -a + \eta_2(2ab - a^3) & 1 + \eta_2(a^2-b)
    \end{bmatrix},
    H_1^{\prime} = \begin{bmatrix}
        -\eta_2 & -\eta_1 \\
        0 & 0
    \end{bmatrix}.\] 
    Given that $\eta_1 = 2 a^{-1} (a^2 - 2b)^{-1} $ and $\eta_2 = - a \eta_1 $, it follows that $G_1^{\prime\prime} = A^{\prime\prime} H_1^{\prime} $. That is, the condition (ii) of Case 5 of Theorem~{\ref{block}} holds. Thus, and by Case 5 of Theorem~{\ref{block}}, $\mathcal{C}_1 = ev_{\boldsymbol{\alpha},\boldsymbol{v}}(\mathcal{S}_1)$ is a $[2rp^s, rp^s]_{q^2}$ self-dual TGRS code. Obviously, $\mathcal{C}_1$ is a subcode of $GRS_{rp^s+2}(\boldsymbol{\alpha},\boldsymbol{v})$. Therefore, $\mathcal{C}_1$ has the minimum distance $d \geq rp^s-1$. It completes the proof.
\end{proof}

\begin{proposition}\label{MDS_block6}
    Let notations be the same as Theorem~{\upshape\ref{block6}}, $ 1 + \eta_1 (c_1 c_2 - c_3) + \eta_2 (c_1^2 - c_2) + \eta_1 \eta_2 c_1 c_2 (c_1^2 + c_2 - \eta_1 c_2 - 1) \neq 0 $. Then there exists a $[2rp^s, rp^s, rp^s+1]_{q^2}$ self-dual MDS TGRS code.
\end{proposition}

\begin{example}
    Let $s=1$, $t=1$ and $r=1$. Let $g(x) = x^{6} + x^{5} + x^{4} + 2 \in \mathbb{F}_{3}[x] $, then by computer, the splitting field of $g(x)$ over $\mathbb{F}_{3}$ is $\mathbb{F}_{3^6} $ and $\mathbb{F}_{3^6}^{*} = \langle \beta \rangle $, where the minimal polynomial of $\beta$ is $x^6 + 2x^4 + x^2 + 2x + 2$. Let $\mathbb{F}_{3^{12}}^{*} = \langle \gamma \rangle $, where the minimal polynomial of $\gamma$ is $x^{12} + x^6 + x^5 + x^4 + x^2 + 2 $. Then the roots of $g(x)$ are $\left\{ 2, \beta^{700}, \beta^{637}, \beta^{476}, \beta^{644}, \beta^{455} \right\} $. Set 
    \[ \boldsymbol{\alpha} = (2, \beta^{700}, \beta^{637}, \beta^{476}, \beta^{644}, \beta^{455}) \] 
    and
    \[    
    \boldsymbol{v} =(1,\gamma^{275940},\gamma^{166075},\gamma^{357700},\gamma^{30660},\gamma^{232505}).
    \]
    Let $\eta_{1} = 1 $ and $\eta_{2} = 2 $, then \[
    \mathcal{S}_1 = \left\{ \sum_{i=0}^{2} f_i x^i + (f_{1} + 2f_2) x^4  \mid f_i \in \mathbb{F}_{3^{12}} \text{ for } 0 \leq i \leq 2 \right\}.
    \]
    According to the Theorem~{\ref{block6}} and Proposition~{\ref{MDS_block6}}, $\mathcal{C}_1 = ev_{\boldsymbol{\alpha},\boldsymbol{v}}(\mathcal{S}_1)$ is a $[6,3,4]_{3^{12}} $ self-dual MDS TGRS code.
\end{example}

\begin{theorem}\label{line3}
    Let $s$, $h$, $t$, $r$ be positive integers, $r$ is odd and $t \mid h$, $s < h$. Let $p$ is an odd prime, $q=p^h$ and such that $q \equiv 1 \pmod 8 $ or $q \equiv 3 \pmod 8 $. Assume that $\mathbb{F}_q$ is the splitting field of $g(x) = x^{rp^s} + \sum\limits_{i=1}^{\frac{rp^s-1}{2}}a_i x^{rp^s-i} + c \in \mathbb{F}_{p^t}[x] $ with $a_i,c \neq 0 $ for $1 \leq i \leq \frac{rp^s-1}{2}$, $\alpha_1, \alpha_2, \cdots, \alpha_{rp^s}$ are all roots of $g(x)$. There exist $v_i \in \mathbb{F}_{q^2}$ such that $v_i^2 = (g'(\alpha_i))^{-1} $ for $1 \leq i \leq {rp^s}$. Let $\boldsymbol{\alpha} = (\alpha_1, \alpha_2, \cdots, \alpha_{rp^s})$, $\boldsymbol{v} = (v_1, v_2, \cdots, v_{rp^s})$. 
    Let $ \mathcal{D} = \left\{ a \in \mathbb{F}_{q}^{*} \mid a^{\frac{q-1}{2}} = 1 \right\} $. Let $ \eta_{\frac{rp^s+1}{2}} \in \mathcal{D} $, $\eta_1 = \eta_{\frac{rp^s+1}{2}} a_{\frac{rp^s-1}{2}} + (-2\eta_{\frac{rp^s+1}{2}})^{\frac{1}{2}} $ and $\eta_i = \eta_{\frac{rp^s+1}{2}} a_{\frac{rp^s+1}{2}-i} $ for $2 \leq i \leq \frac{rp^s-1}{2} $. Set 
    \[
    \begin{aligned}
    \mathcal{S}_2 = \Bigl\{ & \sum_{i=0}^{\frac{rp^s-3}{2}} f_i x^i + f_0 (-2\eta_{\frac{rp^s+1}{2}})^{\frac{1}{2}} x^{\frac{rp^s-1}{2}} + f_0 \eta_{\frac{rp^s+1}{2}} ( \sum_{j=1}^{\frac{rp^s-1}{2}} a_j x^{rp^s-1-j} + x^{rp^s-1}) \\
    &\quad \Bigm| f_i \in \mathbb{F}_{q^2} \text{ for } 0 \leq i \leq \frac{rp^s-3}{2} \Bigr\}.
    \end{aligned}
    \]
    Then $\mathcal{C}_2 = ev_{\boldsymbol{\alpha},\boldsymbol{v}}(\mathcal{S}_2)$ is a $[rp^s, \frac{rp^s-1}{2}]$ self-orthogonal TGRS code over $\mathbb{F}_{q^2}$.
\end{theorem}

\begin{proof}
    Since $g'(x) = x^{rp^s-1} - \sum\limits_{i=1}^{\frac{rp^s-1}{2}} i a_{i} x^{rp^s-i-1} $, we have $\gcd(g(x),g'(x)) = 1$, then $g(x)$ has $rp^s$ distinct roots in $\mathbb{F}_q$, namely, $\alpha_i \neq \alpha_j$ $(i \neq j)$. On the one hand, $u_i = (g'(\alpha_i))^{-1} \in \mathbb{F}_q$ is a square elements of $\mathbb{F}_{q^2}$, then there exist $v_i \in \mathbb{F}_{q^2}$ such that $u_i = v_i^2$, so the condition (i) of Case 1 of Theorem~{\ref{line}} holds.

    On the other hand, write $g(x) = \prod\limits_{i=1}^{rp^s} (x-\alpha_i) = \sum\limits_{i=0}^{rp^s} \sigma_i x^{rp^s-i} $, we have $\sigma_i = a_i$ for $1 \leq i \leq \frac{rp^s-1}{2} $. It is straightforward to verify from the construction of $\eta_{i} $ for $1 \leq i \leq \frac{rp^s+1}{2} $ that $\eta_{\frac{rp^s+1}{2}-i} = \eta_{\frac{rp^s-1}{2}} \sigma_i$ for $1 \leq i \leq \frac{rp^s-3}{2} $ and $\eta_{\frac{rp^s+1}{2}} \sigma_{rp^s-1} + \eta_1 b_{\frac{rp^s-1}{2}} - \eta_{\frac{rp^s+1}{2}} \sigma_{\frac{rp^s-1}{2}} b_{\frac{rp^s-1}{2}} = 2 $. That is, condition (ii) and (iii) of Case 1 of Theorem~{\ref{line}} hold. Thus, and by Case 1 of Theorem~{\ref{line}}, $\mathcal{C}_2 = ev_{\boldsymbol{\alpha},\boldsymbol{v}}(\mathcal{S}_2)$ is a $[rp^s, \frac{rp^s-1}{2}]_{q^2}$ self-orthogonal TGRS code.
\end{proof}

\begin{proposition}\label{MDS_line3}
    Let notations be the same as Theorem~{\upshape\ref{line3}}, $ M(n,k,\boldsymbol{\alpha},A_2(\eta),I) \neq 0 $. Then there exists a $[rp^s, \frac{rp^s-1}{2}, \frac{rp^s+3}{2}]_{q^2}$ self-orthogonal MDS TGRS code.
\end{proposition}

By Lemma~{\ref{QECC}} and Theorem~{\ref{line3}}, we get the following proposition.

\begin{proposition}\label{QECC_line3}
    Let notations be the same as Theorem~{\upshape\ref{line3}}. Then there exists a $[[rp^{s} , 1, wt(\mathcal{C}_2^{\perp} \setminus \mathcal{C}_2)]]_{q^2}$ quantum stabilizer code.
\end{proposition}

\begin{example}
    Let $s=1$, $t=1$ and $r=1$. Let $g(x) = x^{7} + x^{6} + x^{5} + 6x^4 + 2 \in \mathbb{F}_{7}[x] $, then by computer, the splitting field of $g(x)$ over $\mathbb{F}_{7}$ is $\mathbb{F}_{7^{2}} $ and $\mathbb{F}_{7^{2}}^{*} = \langle \beta \rangle $, where the minimal polynomial of $\beta$ is $x^2 + 6x + 3$. Let $\mathbb{F}_{7^{4}}^{*} = \langle \gamma \rangle $, where the minimal polynomial of $\gamma$ is $x^4 + 5x^2 + 4x + 3 $. Then the roots of $g(x)$ are $\left\{ 6, 3, 2, \beta^{12}, \beta^{23}, \beta^{17}, \beta^{36} \right\} $. 
    Set 
    \[ \boldsymbol{\alpha} = (6, 3, 2, \beta^{12}, \beta^{23}, \beta^{17}, \beta^{36}) \] 
    and 
    \[    
    \boldsymbol{v} =(\gamma^{1000}, \gamma^{200}, \gamma^{800}, \gamma^{825}, \gamma^{650}, \gamma^{950}, \gamma^{975}).
    \]
    Let $\eta_1 = \gamma^{64} $ and $\eta_i = \gamma^{492} $ for $2 \leq i \leq 4 $, then \[
    \mathcal{S}_2 = \left\{ \sum_{i=0}^{2} f_i x^i + f_0 ( \gamma^{64} x^3 + \gamma^{492} \sum_{i=4}^6 x^{i})  \mid f_i \in \mathbb{F}_{7^{4}} \text{ for } 0 \leq i \leq 2 \right\}.
    \]
    According to the Theorem~{\ref{line3}} and Proposition~{\ref{MDS_line3}}, $\mathcal{C}_2 = ev_{\boldsymbol{\alpha},\boldsymbol{v}}(\mathcal{S}_2)$ is a $[7,3,5]_{7^{4}} $ self-orthogonal MDS TGRS code and $\mathcal{C}_2^{\perp} $ is a $[7,4,4]_{7^4} $ MDS code. Applying the Proposition~{\ref{QECC_line3}}, we construct a quantum error-correcting code with parameters $[[7,1,4]]_{7^4} $, which are quantum MDS codes achieving the quantum Singleton bound.
\end{example}

\begin{theorem}\label{line4}
    Let $s$, $h$, $t$, $r$ be positive integers, and $t \mid h$, $s < h$. Let $p$ be an odd prime, $q=p^h$. Assume that $a_{2rp^s-1} \in \mathbb{F}_{p^t}^{*} $ and $a_i \in \mathbb{F}_{p^t}^{*}$ for $ 1 \leq i \leq rp^s-1 $ such that the polynomial of the form $g(x) = x^{2rp^s} + \sum\limits_{i=1}^{2rp^s} a_i x^{2rp^s-i} \in \mathbb{F}_{p^t}[x]$ has $2rp^s$ distinct roots in $\mathbb{F}_q$, $\alpha_1, \alpha_2, \cdots, \alpha_{2rp^s}$ are all roots of $g(x)$. There exist $v_i \in \mathbb{F}_{q^2}$ such that $v_i^2 = (g'(\alpha_i))^{-1} $ for $1 \leq i \leq {2rp^s}$. Let $\boldsymbol{\alpha} = (\alpha_1, \alpha_2, \cdots, \alpha_{2rp^s})$, $\boldsymbol{v} = (v_1, v_2, \cdots, v_{2rp^s})$. Let $\eta_{rp^s} = 2 a_{2rp^s-1}^{-1} $ and $\eta_i = \eta_{rp^s} a_{rp^s-i} $ for $1 \leq i \leq rp^s-1 $. Set
    \[
    \mathcal{S}_2 = \left\{ \sum_{i=0}^{rp^s-1} f_i x^i + 2 a_{2rp^s-1}^{-1} f_0 ( \sum_{j=1}^{rp^s-1} a_j x^{2rp^s-1-j} + x^{2rp^s-1})  \mid f_i \in \mathbb{F}_{q^2} \text{ for } 0 \leq i \leq rp^s-1 \right\}.
    \]
    Then $\mathcal{C}_2 = ev_{\boldsymbol{\alpha},\boldsymbol{v}}(\mathcal{S}_2)$ is a $[2rp^s, rp^s]$ self-dual TGRS code over $\mathbb{F}_{q^2}$.
\end{theorem}

\begin{proof}
    On the one hand, $u_i = (g'(\alpha_i))^{-1} \in \mathbb{F}_q$ is a square elements of $\mathbb{F}_{q^2}$, then there exist $v_i \in \mathbb{F}_{q^2}$ such that $u_i = v_i^2$, so the condition (i) of Case 2 of Theorem~{\ref{line}} holds.

    On the other hand, write $g(x) = \prod\limits_{i=1}^{2rp^s} (x-\alpha_i) = \sum\limits_{i=0}^{2rp^s} \sigma_i x^{2rp^s-i} $, we have $\sigma_i = a_i$ for $1 \leq i \leq rp^s-1$ and $\sigma_{2rp^s-1} = a_{2rp^s-1} $. It is straightforward to verify from the construction of $\eta_{i} $ for $1 \leq i \leq rp^s $ that $\eta_{rp^{s}-i} = \eta_{rp^{s}} \sigma_i$ for $1 \leq i \leq rp^{s}-1 $ and $\eta_{rp^{s}} \sigma_{2rp^s-1} = 2 $. That is, condition (ii) and (iii) of Case 2 of Theorem~{\ref{line}} hold. Thus, and by Case 2 of Theorem~{\ref{line}}, $\mathcal{C}_2 = ev_{\boldsymbol{\alpha},\boldsymbol{v}}(\mathcal{S}_2)$ is a $[2rp^s, rp^{s}]_{q^2}$ self-dual TGRS code.
\end{proof}

\begin{proposition}\label{MDS_line4}
    Let notations be the same as Theorem~{\upshape\ref{line4}}, $ M(n,k,\boldsymbol{\alpha},A_2(\eta),I) \neq 0 $. Then there exists a $[2rp^s, rp^s, rp^s+1]_{q^2}$ self-dual MDS TGRS code.
\end{proposition}

\begin{example}
    Let $s=1$, $t=1$ and $r=1$. Let $g(x) = x^6 + x^5 + 2x^4 + 2x^3 + x^2 + x + 2 \in \mathbb{F}_{3}[x] $, then by computer, the splitting field of $g(x)$ over $\mathbb{F}_{3}$ is $\mathbb{F}_{3^{3}} $ and $\mathbb{F}_{3^{3}}^{*} = \langle \beta \rangle $, where the minimal polynomial of $\beta$ is $x^3 + 2x + 1$. Let $\mathbb{F}_{3^{6}}^{*} = \langle \gamma \rangle $, where the minimal polynomial of $\gamma$ is $x^6 + 2x^4 + x^2 + 2x + 2 $. Then the roots of $g(x)$ are $\left\{ \beta,\beta^{3},\beta^{9},\beta^{12},\beta^{4},\beta^{10} \right\} $. 
    Set 
    \[ \boldsymbol{\alpha} = (\beta,\beta^{3},\beta^{9},\beta^{12},\beta^{4},\beta^{10}) \] 
    and 
    \[   
    \boldsymbol{v} =(\gamma^{98},\gamma^{294},\gamma^{154},\gamma^{336},\gamma^{112},\gamma^{280}).
    \]
    Let $\eta_1 = 1 $ and $\eta_2 = \eta_3 = 2 $, then 
    \[
    \mathcal{S}_2 = \left\{ \sum_{i=0}^{2} f_i x^i + f_0 (x^3 + x^4 + x^5)  \mid f_i \in \mathbb{F}_{3^{6}} \text{ for } 0 \leq i \leq 2 \right\}.
    \]
    According to the Theorem~{\ref{line4}} and Proposition~{\ref{MDS_line4}}, $\mathcal{C}_2 = ev_{\boldsymbol{\alpha},\boldsymbol{v}}(\mathcal{S}_2)$ is a $[6,3,4]_{3^{6}} $ self-dual MDS TGRS code.
\end{example}

\section{Conclusion}\label{sec5}

In this paper, the form of the parity check matrices of $\mathcal{C}_1$ and $\mathcal{C}_2$ is determined. Furthermore, the necessary and sufficient conditions under which the TGRS codes $\mathcal{C}_1$ and $\mathcal{C}_2$ are self-orthogonal are given. We also provide several explicit constructions of self-orthogonal TGRS codes, which further yield self-orthogonal AMDS, NMDS, MDS TGRS codes and quantum MDS stabilizer codes that achieve the quantum Singleton bound. 

Recall that $A_1(\eta) = \begin{bmatrix}
    0 & 0 \\
    \Gamma & 0
\end{bmatrix}$, where $\Gamma = \begin{bmatrix}
    \eta_{11} & \eta_{12} \\
    \eta_{21} & \eta_{22}
\end{bmatrix}$. Table~\ref{table} summarizes the existence conditions for the self-orthogonal code $\mathcal{C}_1$ for different choices of $\Gamma$.

\begin{longtable}{p{1.3cm}p{1.3cm}p{3.1cm}p{1.6cm}p{3.6cm}}

\caption{The existence of the self-orthogonal code $\mathcal{C}_1$}\label{table} \\ 

\toprule
$\Gamma$ & $n$ & conditions & self-orthogonal & reference\\
\midrule
\endfirsthead

\toprule
$\Gamma$ & $n$ & conditions & self-orthogonal & reference\\
\midrule
\endhead

\bottomrule
\endfoot

\bottomrule
\endlastfoot

    \multirow{3}{*}[-1ex]{$\begin{bmatrix}
        \eta & 0  \\
        0 & 0 
    \end{bmatrix}$} & $\geq 2k+2$ & & $\checkmark$ & Th 6, \cite{liu2021construction} \\
    \cmidrule{2-5}
     & $2k+1$ &  & $\times$ & Th~{\ref{block}}, Case 4 \\
    \cmidrule{2-5}
     & $2k$ &  & $\times$ & Cor 5.2, \cite{sui2023new} \\
     
     \midrule
     \multirow{5}{*}[-2ex]{$\begin{bmatrix}
         0 & \eta \\
         0 & 0
     \end{bmatrix}$} & $\geq 2k+4$ & & $\checkmark$ & Th~{\ref{block}}, Case 1 \\
    \cmidrule{2-5}
      & $2k+3$ & & $\times$ & Th~{\ref{block}}, Case 2\\
      \cmidrule{2-5}
      & $2k+2$ & $\sigma_1 = 0$ & $\checkmark$ & Th~{\ref{block}}, Case 3\\
      \cmidrule{2-5}
      & $2k+1$ & & $\times$ & Th~{\ref{block}}, Case 4\\
      \cmidrule{2-5}
      & $2k$ & \begin{tabular}[c]{@{}c@{}}$\sigma_1 = 0$\\ $\eta \sigma_3 = 2$\end{tabular} & $\checkmark$ & Cor 5.3, \cite{sui2023new} \\
      
      \multirow{3}{*}[-1ex]{$\begin{bmatrix}
        0 & 0  \\
        \eta & 0 
    \end{bmatrix}$} & $\geq 2k+2$ & & $\checkmark$ & Th 6, \cite{liu2021construction} \\
    \cmidrule{2-5}
      & $2k+1$ &  & $\times$ & Th 4.1 (4), \cite{liang2025multi} \\
    \cmidrule{2-5}
      & $2k$ & $\eta \sigma_1 = 2$ & $\checkmark$ & \begin{tabular}[c]{@{}l@{}}Th 4.1 (1)(2)(3), 4.2 \cite{liang2025multi}\\ Th 2.8, \cite{huang2021mds}\end{tabular} \\
      
      \midrule
      \multirow{5}{*}[-2ex]{$\begin{bmatrix}
         0 & 0 \\
         0 & \eta
     \end{bmatrix}$} & $\geq 2k+4$ & & $\checkmark$ & Th~{\ref{block}}, Case 1\\
     \cmidrule{2-5}
      & $2k+3$ & & $\times$ & Cor~{\ref{block_single}}\\
      \cmidrule{2-5}
      & $2k+2$ & $\sigma_1 = 0$ & $\checkmark$ & Th 4.4 (2), \cite{zhang2025almost} \\
      \cmidrule{2-5}
      & $2k+1$ & $\eta \sigma_2 - \eta \sigma_1^2 = 2$ & $\checkmark$ & Th 4.6, \cite{zhang2025almost} \\
      \cmidrule{2-5}
      & $2k$ &  & $\times$ & \begin{tabular}[c]{@{}c@{}}Th 4.3, \cite{zhang2025almost}\\  Cor 5.2, \cite{sui2023new}\end{tabular}  \\
      
      \midrule
      \multirow{3}{*}[-2ex]{$\begin{bmatrix}
          \eta_1 & \eta_2 \\
          0 & 0
      \end{bmatrix}$} & $\geq 2k+4$ & & $\checkmark$ & Th~{\ref{block}}, Case 1\\
      \cmidrule{2-5}
      & $2k+3$ & & $\times$ & Th~{\ref{block}}, Case 2\\
      \cmidrule{2-5}
      & $2k+2$ & $\eta_2 \sigma_1 = 2 \eta_1$ & $\checkmark$ & Th~{\ref{block}}, Case 3\\
      
      \pagebreak
      & $2k+1$ & & $\times$ & Th~{\ref{block}}, Case 4\\
      \cmidrule{2-5}
      & $2k$ & \begin{tabular}[c]{@{}c@{}}$\eta_2 \sigma_1 = \eta_1$\\ $\eta_2 \sigma_3 = 2$\end{tabular} & $\checkmark$ & \begin{tabular}[c]{@{}c@{}}Cor 5.8, \cite{sui2023new}\\ Cor 4.3, \cite{ding2025new}\end{tabular} \\ 
      
      \midrule
      \multirow{5}{*}[-2ex]{$\begin{bmatrix}
          0 & 0 \\
          \eta_1 & \eta_2 
      \end{bmatrix}$} & $\geq 2k+4$ & & $\checkmark$ & Th 4.6, \cite{liang2025multi} \\
      \cmidrule{2-5}
      & $2k+3$ & & $\times$ & Th 4.8 (3), \cite{liang2025multi}\\
      \cmidrule{2-5}
      & $2k+2$ & $\eta_2 \sigma_1 = 2 \eta_1$ & $\checkmark$ & Th 4.5, \cite{liang2025multi} \\
      \cmidrule{2-5}
      & \multirow{2}{*}{$2k+1$} & $2 \eta_1 \eta_2 \sigma_1 + \eta_2^2 \sigma_2 - \eta_2^2 \sigma_1 -2 \eta_2 - \eta_1^2 = 0$ & $\checkmark$ & \multirow{2}{*}{Th 4.3, 4.4, \cite{liang2025multi}}\\
      \cmidrule{2-5}
      & $2k$ & & $\times$ & \begin{tabular}[c]{@{}l@{}}Cor 5.11, \cite{sui2023new}\\ Th 4.8 (1), \cite{liang2025multi}\end{tabular} \\
      
      \midrule
      \multirow{3}{*}[-1ex]{$\begin{bmatrix}
        \eta_1 & 0  \\
        \eta_2 & 0 
      \end{bmatrix}$} & $\geq 2k+2$ & & $\checkmark$ & Th~{\ref{block}}, Case 1, 2, 3\\
      \cmidrule{2-5}
      & $2k+1$ &  & $\times$ & Th~{\ref{block}}, Case 4\\
     \cmidrule{2-5}
     & $2k$ &  & $\times$ & Cor 5.11, \cite{sui2023new} \\
     
     \midrule
    \multirow{5}{*}[-2ex]{$\begin{bmatrix}
          0 & \eta_1 \\
          0 & \eta_2  
      \end{bmatrix}$} & $\geq 2k+4$ & & $\checkmark$ & Th~{\ref{block}}, Case 1 \\
      \cmidrule{2-5}
      & $2k+3$ & & $\times$ & Th~{\ref{block}}, Case 2\\
      \cmidrule{2-5}
      & $2k+2$ & $\sigma_1 = 0$ & $\checkmark$ & Th~{\ref{block}}, Case 3 \\
      \cmidrule{2-5}
      & $2k+1$ &  & $\times$ & Th~{\ref{block}}, Case 4\\
      \cmidrule{2-5}
      & $2k$ & \parbox{3.1cm}{$\eta_2 = - \eta_1 \sigma_1$\\ $\eta_1 \sigma_1^3  - 2 \eta_1 \sigma_1 \sigma_2 + \eta_1 \sigma_3= 2$} & $\checkmark$ & Cor 5.16, \cite{sui2023new} \\

      \midrule
      \multirow{5}{*}[-2ex]{$\begin{bmatrix}
          \eta_1 & 0 \\
          0 & \eta_2 
      \end{bmatrix}$} & $\geq 2k+4$ & & $\checkmark$ & Th~{\ref{block}}, Case 1 \\
      \cmidrule{2-5}
      & $2k+3$ & & $\times$ & Th~{\ref{block}}, Case 2\\
      \cmidrule{2-5}
      & $2k+2$ &  & $\times$ & Cor~{\ref{block_special}} \\
      \cmidrule{2-5}
      & $2k+1$ &  & $\times$ & Th~{\ref{block}}, Case 4\\
      \cmidrule{2-5}
      & $2k$ & \begin{tabular}[c]{@{}l@{}}$\sigma_1 = \sigma_3 = 0$\\ $\eta_1 \eta_2 \sigma_2 = \eta_1 + \eta_2 $\end{tabular} & $\checkmark$ & \begin{tabular}[c]{@{}l@{}}Cor 5.12, \cite{sui2023new}\\ Th 5.2, \cite{gu2023twisted}\end{tabular} \\

      \midrule
      \multirow{5}{*}[-2ex]{$\begin{bmatrix}
          0 & \eta_2 \\
          \eta_1 & 0 
      \end{bmatrix}$} & $\geq 2k+4$ & & $\checkmark$ & Th~{\ref{block}}, Case 1 \\
      \cmidrule{2-5}
      & $2k+3$ & & $\times$ & Th~{\ref{block}}, Case 2\\
      \cmidrule{2-5}
      & $2k+2$ &  & $\times$ & Cor~{\ref{block_special}} \\
      \cmidrule{2-5}
      & $2k+1$ &  & $\times$ & Th~{\ref{block}}, Case 4\\
      \cmidrule{2-5}
      & $2k$ & \begin{tabular}[c]{@{}l@{}}$\eta_1 \sigma_1 = 2 $\\ $\eta_2 \sigma_3 = 2$\\ $\eta_1 \sigma_2 =  \sigma_1$\end{tabular} & $\checkmark$ & \begin{tabular}[c]{@{}l@{}}Cor 5.15, \cite{sui2023new}\\ Th 5.1, \cite{sui2022mds}\end{tabular} \\
      \pagebreak

      \multirow{5}{*}[-2ex]{$\begin{bmatrix}
          \eta_{11} & \eta_{12} \\
          \eta_{21} & \eta_{22} 
      \end{bmatrix}$} & $\geq 2k+4$ & & $\checkmark$ & Th~{\ref{block}}, Case 1 \\
      \cmidrule{2-5}
      & $2k+3$ & $\eta_{12} = \eta_{22} = 0$ & $\checkmark$ & Th~{\ref{block}}, Case 2\\
      \cmidrule{2-5}
      & $2k+2$ & \parbox{3.1cm}{$\eta_{i2}(\eta_{i2} \sigma_1 - 2 \eta_{i1}) = 0 $\\ $  \eta_{11}\eta_{22} + \eta_{12}\eta_{21} = \eta_{12} \eta_{22} \sigma_1 $} & $\checkmark$ & Th~{\ref{block}}, Case 3 \\
      \cmidrule{2-5}
      & $2k+1$ & $G_1^{\prime} = A^{\prime} H_1^{\prime} $ & $\checkmark$ & Th~{\ref{block}}, Case 4\\
      \cmidrule{2-5}
      & $2k$ & $G_1^{\prime\prime} = A^{\prime\prime} H_1^{\prime} $ & $\checkmark$ & \begin{tabular}[c]{@{}l@{}}Th~{\ref{block}}, Case 4\\ Th 5.1, \cite{sui2023new}\end{tabular} \\
\end{longtable}

Table~{\ref{table}} summarizes all corollaries derived from Theorem~{\ref{block}} and covers various special cases of $\Gamma$ presented in previous studies. It follows that our work generalizes the relevant prior results. In future work, we should use the similar methods to derive the sufficient and necessary conditions for codes with $\Gamma = \begin{bmatrix}
    \eta_{11} & \eta_{12} & \cdots & \eta_{1l} \\
    \eta_{21} & \eta_{22} & \cdots & \eta_{2l}
\end{bmatrix}$ to be self-orthogonal.

\section*{Statements and Declarations}

\begin{itemize}
\item \textbf{Funding} This work was supported by the Fundamental Research Funds for the Central Universities (26CX03010A), China Education Innovation Research Fund (2022BL027), Shandong Provincial Natural Science Foundation of China (ZR2023LLZ013) and the Key Project of Computing Power Internet and Information Security, Ministry of Education (2024ZD006).
\item \textbf{Competing interests} The authors declare no competing interests.
\item \textbf{Data availability} No datasets were generated or analyzed during the current study.
\item \textbf{Author contribution} Authors Yanxin Chen and Yanli Wang were primarily responsible for deriving main results and manuscript writing. Tongjiang Yan contributed to the development and verification of the proofs and offered guidance throughout the project.
\end{itemize}

\bibliography{main}

@article{carlet2018euclidean,
  title={{Euclidean and Hermitian LCD MDS codes}},
  author={Carlet, Claude and Mesnager, Sihem and Tang, Chunming and Qi, Yanfeng},
  journal={Designs, Codes and Cryptography},
  volume={86},
  number={11},
  pages={2605--2618},
  year={2018},
  publisher={Springer}
}

@article{steane2002enlargement,
  title={{Enlargement of calderbank-shor-steane quantum codes}},
  author={Steane, Andrew M},
  journal={IEEE Transactions on Information Theory},
  volume={45},
  number={7},
  pages={2492--2495},
  year={2002},
  publisher={IEEE}
}

@article{ketkar2006nonbinary,
  title={{Nonbinary stabilizer codes over finite fields}},
  author={Ketkar, Avanti and Klappenecker, Andreas and Kumar, Santosh and Sarvepalli, Pradeep Kiran},
  journal={IEEE transactions on information theory},
  volume={52},
  number={11},
  pages={4892--4914},
  year={2006},
  publisher={IEEE}
}

@article{calderbank1997quantum,
  title={{Quantum error correction and orthogonal geometry}},
  author={Calderbank, A Robert and Rains, Eric M and Shor, Peter W and Sloane, Neil JA},
  journal={Physical Review Letters},
  volume={78},
  number={3},
  pages={405},
  year={1997},
  publisher={APS}
}

@article{gottesman1996class,
  title={{Class of quantum error-correcting codes saturating the quantum Hamming bound}},
  author={Gottesman, Daniel},
  journal={Physical Review A},
  volume={54},
  number={3},
  pages={1862},
  year={1996},
  publisher={APS}
}

@book{huffman2010fundamentals,
  title={{Fundamentals of error-correcting codes}},
  author={Huffman, W Cary and Pless, Vera},
  publisher={Cambridge university press}
}

@book{macwilliams1977theory,
  title={{The theory of error-correcting codes}},
  author={MacWilliams, Florence Jessie and Sloane, Neil James Alexander},
  volume={16},
  publisher={Elsevier}
}

@article{baicheva200410,
  title={{On the [10, 5, 6] Reed-Solomon and Glynn Codes}},
  author={Baicheva, T and Bouyukliev, I and Dodunekov, S and Willems, W},
  journal={Update},
  volume={2004},
  pages={01--14},
  year={2004}
}

@article{georgiou2002mds,
  title={{MDS self-dual codes over large prime fields}},
  author={Georgiou, S and Koukouvinos, C},
  journal={Finite Fields and Their Applications},
  volume={8},
  number={4},
  pages={455--470},
  year={2002},
  publisher={Elsevier}
}

@inproceedings{grassl2008self,
  title={{On self-dual MDS codes}},
  author={Grassl, Markus and Gulliver, T Aaron},
  booktitle={2008 IEEE International Symposium on Information Theory},
  pages={1954--1957},
  year={2008},
  organization={IEEE}
}

@article{guenda2012new,
  title={{New MDS self-dual codes over finite fields}},
  author={Guenda, Kenza},
  journal={Designs, Codes and Cryptography},
  volume={62},
  number={1},
  pages={31--42},
  year={2012},
  publisher={Springer}
}

@article{kim2004euclidean,
  title={{Euclidean and Hermitian self-dual MDS codes over large finite fields}},
  author={Kim, Jon-Lark and Lee, Yoonjin},
  journal={Journal of combinatorial theory, series A},
  volume={105},
  number={1},
  pages={79--95},
  year={2004},
  publisher={Elsevier}
}

@article{shi2018self,
  title={{Self-dual codes and orthogonal matrices over large finite fields}},
  author={Shi, Minjia and Sok, Lin and Sol{\'e}, Patrick and {\c{C}}alkavur, Selda},
  journal={Finite Fields and Their Applications},
  volume={54},
  pages={297--314},
  year={2018},
  publisher={Elsevier}
}

@article{sheekey2015new,
  title={{A new family of linear maximum rank distance codes}},
  author={Sheekey, John},
  journal={arXiv preprint arXiv:1504.01581},
  year={2015}
}

@inproceedings{beelen2017twisted,
  title={{Twisted reed-solomon codes}},
  author={Beelen, Peter and Puchinger, Sven and n{\'e} Nielsen, Johan Rosenkilde},
  booktitle={2017 IEEE International Symposium on Information Theory (ISIT)},
  pages={336--340},
  year={2017},
  organization={IEEE}
}

@article{ashikhmin2002nonbinary,
  title={{Nonbinary quantum stabilizer codes}},
  author={Ashikhmin, Alexei and Knill, Emanuel},
  journal={IEEE Transactions on Information Theory},
  volume={47},
  number={7},
  pages={3065--3072},
  year={2002},
  publisher={IEEE}
}

@article{zhao2025research,
  title={{Research on the construction of maximum distance separable codes via arbitrary twisted generalized Reed-Solomon codes}},
  author={Zhao, Cune’e and Ma, Wenping and Yan, Tongjiang and Sun, Yuhua},
  journal={IEEE Transactions on Information Theory},
  year={2025},
  publisher={IEEE}
}

@article{gu2023twisted,
  title={{On twisted generalized Reed-Solomon codes with l twists}},
  author={Gu, Haojie and Zhang, Jun},
  journal={IEEE Transactions on Information Theory},
  volume={70},
  number={1},
  pages={145--153},
  year={2023},
  publisher={IEEE}
}

@article{zhu20241,
  title={{The [1, 0]-twisted generalized Reed-Solomon code}},
  author={Zhu, Canze and Liao, Qunying},
  journal={Cryptography and Communications},
  volume={16},
  number={4},
  pages={857--878},
  year={2024},
  publisher={Springer}
}

@article{huang2023mds,
  title={{MDS or NMDS LCD codes from twisted Reed-Solomon codes}},
  author={Huang, Daitao and Yue, Qin and Niu, Yongfeng},
  journal={Cryptography and communications},
  volume={15},
  number={2},
  pages={221--237},
  year={2023},
  publisher={Springer}
}

@article{sui2023new,
  title={{New constructions of self-dual codes via twisted generalized Reed-Solomon codes}},
  author={Sui, Junzhen and Yue, Qin and Sun, Fuqing},
  journal={Cryptography and communications},
  volume={15},
  number={5},
  pages={959--978},
  year={2023},
  publisher={Springer}
}

@article{sui2022mds,
  title={{MDS, near-MDS or 2-MDS self-dual codes via twisted generalized Reed-Solomon codes}},
  author={Sui, Junzhen and Yue, Qin and Li, Xia and Huang, Daitao},
  journal={IEEE Transactions on Information Theory},
  volume={68},
  number={12},
  pages={7832--7841},
  year={2022},
  publisher={IEEE}
}

@article{zhu2021self,
  title={{Self-dual twisted generalized Reed-Solomon codes}},
  author={Zhu, Canze and Liao, Qunying},
  journal={arXiv preprint arXiv:2111.11901},
  year={2021}
}

@article{zhu2022self,
  title={{Self-orthogonal generalized twisted Reed-Solomon codes}},
  author={Zhu, Canze and Liao, Qunying},
  journal={arXiv preprint arXiv:2201.02758},
  year={2022}
}

@article{liu2021construction,
  title={{Construction of MDS twisted Reed--Solomon codes and LCD MDS codes}},
  author={Liu, Hongwei and Liu, Shengwei},
  journal={Designs, Codes and Cryptography},
  volume={89},
  number={9},
  pages={2051--2065},
  year={2021},
  publisher={Springer}
}

@article{huang2021mds,
  title={{MDS or NMDS self-dual codes from twisted generalized Reed--Solomon codes}},
  author={Huang, Daitao and Yue, Qin and Niu, Yongfeng and Li, Xia},
  journal={Designs, Codes and Cryptography},
  volume={89},
  number={9},
  pages={2195--2209},
  year={2021},
  publisher={Springer}
}

@article{ding2025new,
  title={{New self-dual codes from TGRS codes with general l twists}},
  author={Ding, Yun and Zhu, Shixin},
  journal={Advances in Mathematics of Communications},
  volume={19},
  number={2},
  pages={662--675},
  year={2025},
  publisher={Advances in Mathematics of Communications}
}

@article{zhu2024class,
  title={{A class of double-twisted generalized Reed-Solomon codes}},
  author={Zhu, Canze and Liao, Qunying},
  journal={Finite Fields and Their Applications},
  volume={95},
  pages={102395},
  year={2024},
  publisher={Elsevier}
}

@article{singh2024mds,
  title={{MDS multi-twisted Reed-Solomon codes with small dimensional hull}},
  author={Singh, Harshdeep and Meena, Kapish Chand},
  journal={Cryptography and Communications},
  volume={16},
  number={3},
  pages={557--578},
  year={2024},
  publisher={Springer}
}

@article{meena2025class,
  title={{A class of triple-twisted GRS codes}},
  author={Meena, Kapish Chand and Pachauri, Piyush and Awasthi, Ambrish and Bhaintwal, Maheshanand},
  journal={Designs, Codes and Cryptography},
  volume={93},
  number={7},
  pages={2369--2393},
  year={2025},
  publisher={Springer}
}

@article{hu2025p,
  title={{On (L,P)-Twisted Generalized Reed-Solomon Codes}},
  author={Hu, Zhao and Wang, Liang and Li, Nian and Zeng, Xiangyong and Tang, Xiaohu},
  journal={IEEE Transactions on Information Theory},
  year={2025},
  publisher={IEEE}
}

@article{liang2025four,
  title={{Four classes of LCD codes from (*)-(L, P)-twisted generalized Reed-Solomon codes}},
  author={Liang, Zhonghao and Liao, Qunying},
  journal={arXiv preprint arXiv:2509.14878},
  year={2025}
}

@article{zhao2025hermitian,
  title={{Hermitian Self-dual Twisted Generalized Reed-Solomon Codes}},
  author={Zhao, Chun'e and Han, Yuxin and Ma, Wenping and Yan, Tongjiang and Sun, Yuhua},
  journal={arXiv preprint arXiv:2508.09687},
  year={2025}
}

@article{guo2023duality,
  title={{Duality of generalized twisted Reed-Solomon codes and Hermitian self-dual MDS or NMDS codes}},
  author={Guo, Guanmin and Li, Ruihu and Liu, Yang and Song, Hao},
  journal={Cryptography and Communications},
  volume={15},
  number={2},
  pages={383--395},
  year={2023},
  publisher={Springer}
}

@article{zhang2025almost,
  title={{Almost self-dual MDS codes and NMDS codes from twisted generalized Reed--Solomon codes}},
  author={Zhang, Yuqi and Ding, Yang},
  journal={Journal of Algebra and Its Applications},
  pages={2650227},
  year={2025},
  publisher={World Scientific}
}

@article{yang2025two,
  title={{Two classes of twisted generalized Reed-Solomon codes with two twists}},
  author={Yang, Shudi and Wang, Jinlong and Wu, Yansheng},
  journal={Finite Fields and Their Applications},
  volume={104},
  pages={102595},
  year={2025},
  publisher={Elsevier}
}

@article{luo2022two,
  title={{Two new classes of Hermitian self-orthogonal non-GRS MDS codes and their applications}},
  author={Luo, Gaojun and Cao, Xiwang and Ezerman, Martianus Frederic and Ling, San},
  journal={Advances in Mathematics of Communications},
  volume={16},
  number={4},
  pages={921--933},
  year={2022},
  publisher={Advances in Mathematics of Communications}
}

@article{sui2022onlymds,
  title={{MDS and near-MDS codes via twisted Reed--Solomon codes}},
  author={Sui, Junzhen and Zhu, Xiaomeng and Shi, Xueying},
  journal={Designs, Codes and Cryptography},
  volume={90},
  number={8},
  pages={1937--1958},
  year={2022},
  publisher={Springer}
}

@article{liang2025multi,
  title={Multi-Twisted Generalized Reed-Solomon Codes: Structure, Properties, and Constructions},
  author={Liang, Zhonghao and Jia, Chenlu and Huang, Dongmei and Liao, Qunying and Tang, Chunming},
  journal={arXiv preprint arXiv:2511.03398},
  year={2025}
}

\end{document}